\documentclass[12pt]{article}
\usepackage[round,sort,authoryear]{natbib}

\usepackage{amsmath,amssymb,amsthm,amsfonts, fullpage, setspace}

\usepackage{eucal}
\usepackage{subeqnarray,bm}
 \usepackage[utf8]{inputenc}
\usepackage{graphicx}[final]
\usepackage{psfrag} %only include if using pictures
\usepackage{color}
\usepackage{pdfpages}

\usepackage{epstopdf}

\usepackage[colorlinks]{hyperref}

\usepackage{hyperref}
\hypersetup{colorlinks=red,citecolor=blue,pdfstartview=FitH, pdfpagemode=None}
\usepackage{multirow}
\usepackage[bottom]{footmisc}
\usepackage{ifthen} %only include if using conditional package
\usepackage{sectsty}
\usepackage{subfig}
\usepackage{float}
\usepackage[english]{babel}
\usepackage[nottoc]{tocbibind}
\usepackage{calrsfs}
\usepackage{soul}
 \DeclareGraphicsExtensions{.eps,.png,.pdf}
\usepackage{mwe}
\makeatletter
\def\maxwidth{ %
  \ifdim\Gin@nat@width>\linewidth
    \linewidth
  \else
    \Gin@nat@width
  \fi
}
\makeatother

\definecolor{fgcolor}{rgb}{0.345, 0.345, 0.345}

\allsectionsfont{\normalsize\bfseries}

\usepackage{framed}
\makeatletter
 {\par\unskip\endMakeFramed%
 \at@end@of@kframe}
\makeatother

\definecolor{shadecolor}{rgb}{.97, .97, .97}
\definecolor{messagecolor}{rgb}{0, 0, 0}
\definecolor{warningcolor}{rgb}{1, 0, 1}
\definecolor{errorcolor}{rgb}{1, 0, 0}
\newcommand{\Intv}[1]{\left[#1\right]}
\newtheorem{thm}{Theorem}

\newtheorem{defn}[thm]{Definition}

\newtheorem{lemma}[thm]{Lemma}

\newtheorem{proposition}[thm]{Proposition}

\newtheorem{remark}[thm]{Remark}

\newcommand{\ds}{\displaystyle}

\newcommand{\norm}[1]{\left\Vert#1\right\Vert}
\newcommand{\abs}[1]{\left\vert#1\right\vert}
\newcommand{\set}[1]{\left\{#1\right\}}
\newcommand{\scp}[1]{\left\langle#1\right\rangle}

\newcommand{\rb}[1]{\left(#1\right)}

\newcommand{\R}{\mathbb{R}}

\newcommand{\N}{\mathbb{N}}

\numberwithin{equation}{section}

\usepackage{calrsfs}

\allsectionsfont{\normalsize\bfseries}

\numberwithin{equation}{section}
\allowdisplaybreaks

\DeclareMathAlphabet{\pazocal}{OMS}{zplm}{m}{n}
\begin{document}
\title{
Hierarchical Model with Allee Effect, Immigration, and Holling Type II Functional Response
}
\author{Eddy Kwessi\footnote{Corresponding author: Department of Mathematics, Trinity University, 1 Trinity Place, San Antonio, TX 78212, Email: ekwessi@trinity.edu}}
\date{}
\maketitle
\begin{abstract}
In this paper, we discuss a hierarchical model, based on a Ricker competition model. The species considered are competing for resources and may be subject to an Allee effect due to mate limitation, anti-predator vigilance or aggression, cooperative predation or resource defense, or social thermoregulation. The species may be classified into a more dominant species  and less dominant or ``wimpy" species or just as a predator and prey. 
The model under consideration also has components taking into account immigration in both  species  and a more structured Holling type II functional response. Local and global stability analyses are  discussed and simulations are provided. We also consider demographic stochasticity on the species 
due to environmental fluctuations in the form of Wiener processes and we show that there are conditions under which a global solution exists, stationary distributions exists,  and strong persistence in mean of the species is possible. We also use Wasserstein distance to show empirically that stochasticity can acts as bifurcation parameter. 
\end{abstract}
%%%%%%%%%%%%%%%%%%%%%%%%%%%%%%%%%%%%%%%%%%%%%%%%%%%%%%%%%%%%%%%%%%%%
\section{Introduction}
 %%%%%%%%%%%%%%%%%%%%%%%%%%%%%%%%%%%%%%%%%%%%%%%%%%%%%%%%%%%%%%%%%%%%
Complex systems in nature often have an either open or hidden hierarchy between its parts or subparts. Therefore dissections of  its structure can help understand its dynamics. This hierarchy may be dictated by size, individual strength, group structure or organization.
For instance, in a bee colony, there is queen who is the largest member of the colony by size, males drones,  and under-developed female called  workers. Another example can be found among different species living in the wild, for instance, there is a hierarchy in strength between lions, hyenas, and say antelopes. In the wild especially, there is a constant displacement of species, either individually or by group, or  due to research for food or water, due changes in the environment, or due to mating needs. There is a vast literature on using hierarchical models in biological systems. We will focus in this manuscript on populations dynamics in the ecological world. In particular, we are interested studying the dynamics between species that are subject to immigration and emigration, subject to an Allee effect, with a more structured functional response. For self-containment, we recall that an Allee effects (see \cite{Allee1949}) is a phenomenon  in population dynamics where there  is a  positive correlation between a population density and its  relative growth rate. It is sometime divided into weak and strong Allee effect, see for instance \cite{Hutchings2015}. The strong Allee effect occurs when a population has a critical density $A$ below which it declines to extinction and above which it increases towards its carrying capacity $K$. The weak Allee effect occurs when a population lacks such a critical density, but at lower densities, the population growth rate arises with increasing densities. \\
\begin{figure}[H] 
   \centering\includegraphics[scale=0.5]{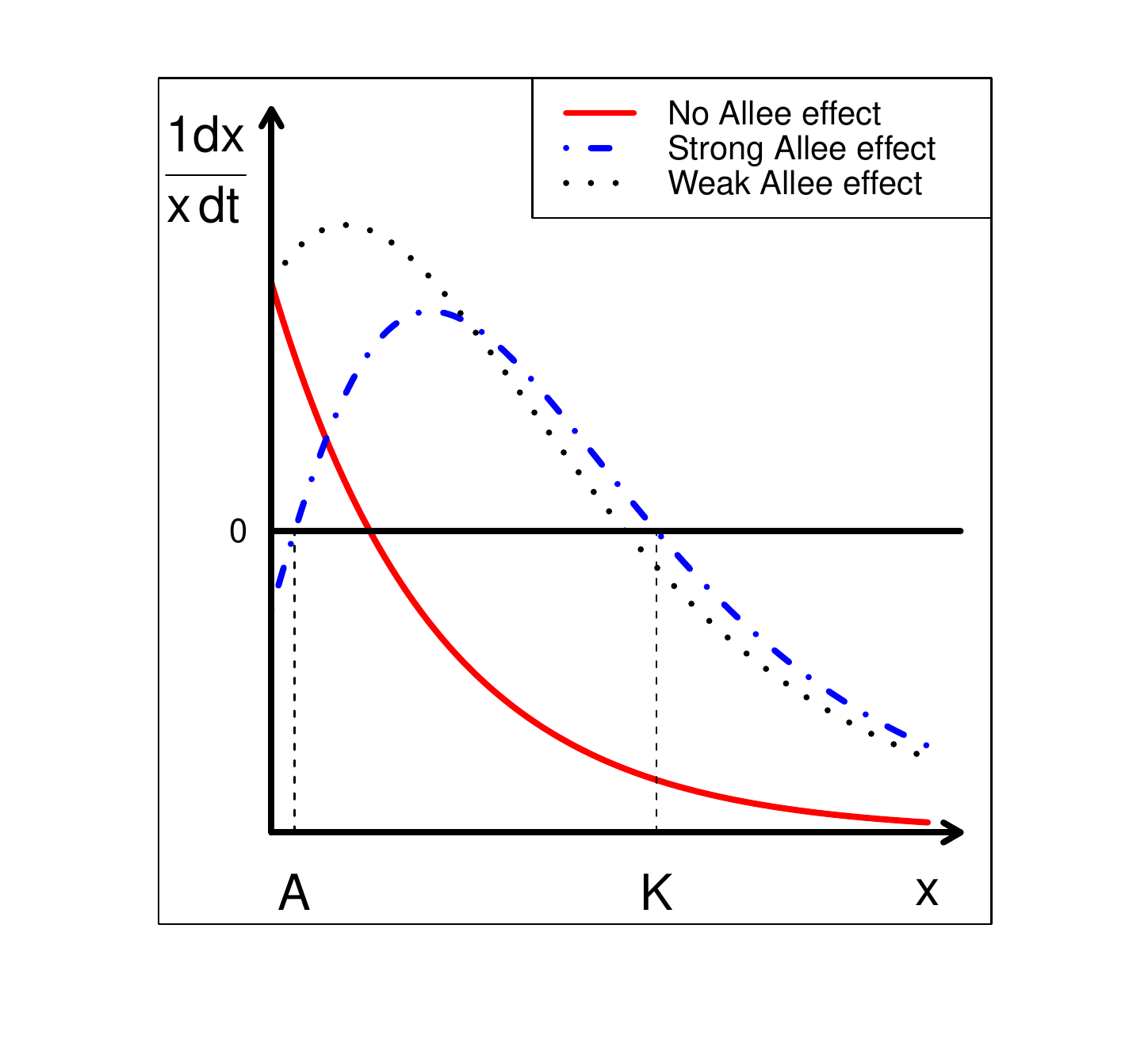}     
   \caption{ Represented is  the relative capita growth rate   as a function of the  population  density  $x$. The red curve represents no Allee effect since there is a negative correlation between the per capita growth rate and the density. The black dashed curve represents the weak Allee effect since at lower densities $x$, the per capita growth rate is increasing but there is no Allee threshold $A$. 
    The blue dashed represents  the strong Allee effect since there is positive correlation between the per capita growth rate and the existence of a Threshold under with the population decreases to extinction.}
   \label{fig:AlleeEffect1}
\end{figure}
We also recall in that in ecology, a functional response represents the intake rate of a species as a function of food density. \cite{Holling1959} proposed  three types of functional responses. Let $x$ and $y$ represents the density of two populations; then a type I functional response is of the form $h_1(x,y)=\mu xy$, where $\mu$ is a constant; a type II functional response is of the form $\ds h_1(x,y)=\frac{\alpha xy}{1+\alpha \beta xy}$, where $\alpha$ represents the attack rate and $\beta$ is the handling time, that is, the time spent by say species $x$ searching and  processing food obtained from species $y$. A type III functional response is of the form $\ds h_1(x,y)=\frac{\mu}{1+\alpha e^{-\beta xy}}$, where $\alpha$ and $\beta$ are as above and $\mu$ is a constant representing a saturation level, that is, a rate threshold when species densities are high. 
\begin{figure}[H] %
   \centering\includegraphics[scale=0.4]{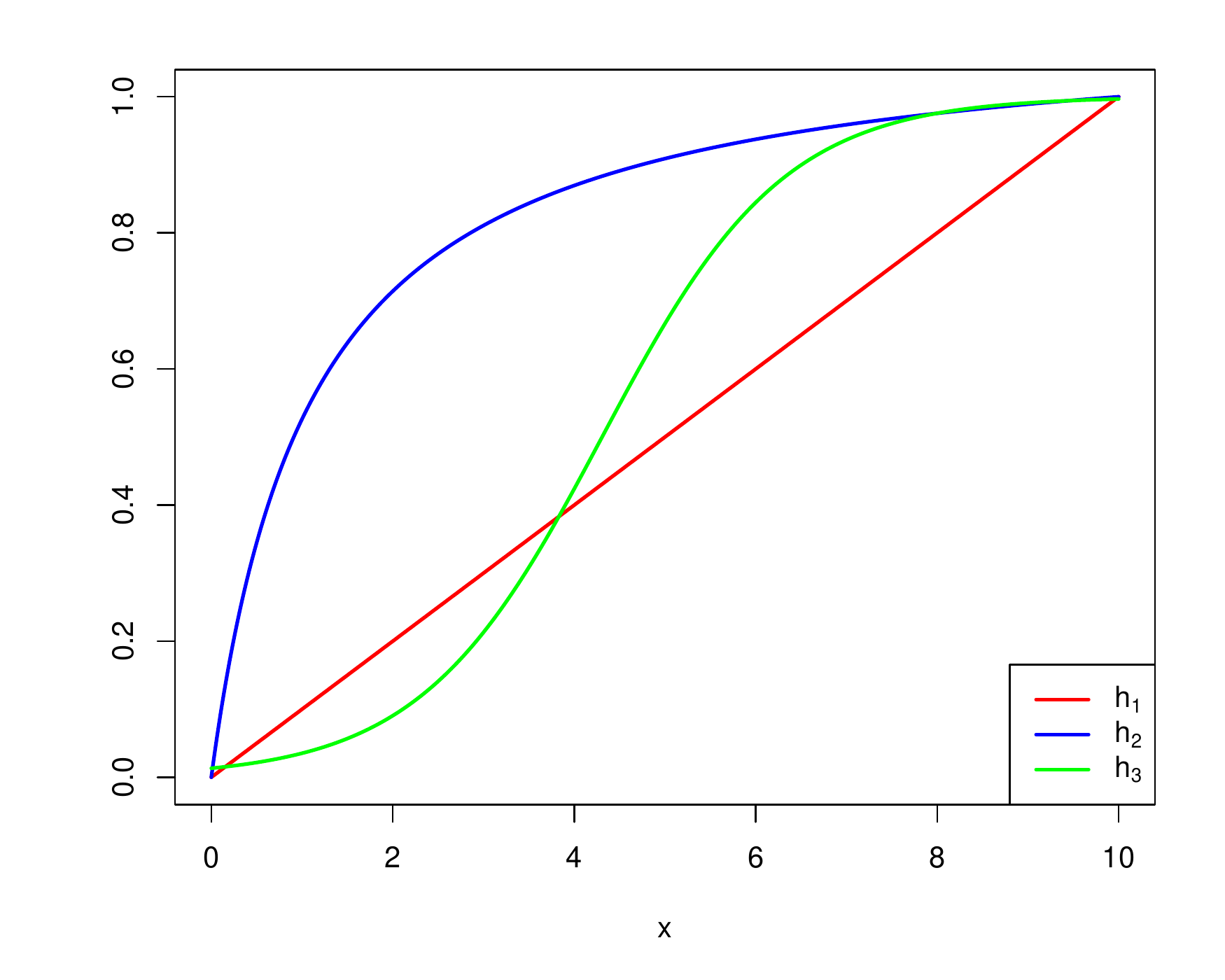}
   \caption{Represented are the three functional response types with chosen parameters. }
   \label{fig:Holling2D}
\end{figure}
In ecology,  there is already a vast literature on the stability analysis of deterministic models. For hierarchical model with one, two, or more species, the interested reader can see for instance \cite{Kwessi2015_7}, \cite{Kwessi2015_6}, and \cite{Kwessi2018_4}. Considering that species live in  habitats that are often subject to demographics fluctuations or perturbations, it is sometime more accurate to consider stochastic  models. In this case, persistence and coexistence of species despite environmental fluctuations are  of particular interest. Papers such as \cite{Chesson2009}, \cite{Benaim2009},  \cite{Hening2021}, \cite{Hening2022},  the references therein, and their subsequent iterations  are great introduction to the  understanding of  the biological motivations and necessary  theoretical underpinnings.  The remainder of the paper if organized as follows:  In section \ref{sect:DeterMdel}, we discuss the deterministic model by finding its fixed points, then we discuss its stability both local and global , and we propose some numerical results. In section \ref{sect:stochatic}, we discuss the stochastic model. In particular, we show existence of a global solution,  we show strong persistence in mean, then we proved the existence of  a stationary distribution, and we propose numerical results. These numerical results consist of phase space diagrams, histograms, and level curves of species densities as well a comparative analysis based on the Wasserstein distance. In section \ref{sect:discussion}, we make some concluding remarks. 
%%%%%%%%%%%%%%%%%%%%%%%%%%%%%%%%%%%%%%%%%%%%%%%%%%%%%%%%%%%%
\section{Deterministic Model} \label{sect:DeterMdel}
%%%%%%%%%%%%%%%%%%%%%%%%%%%%%%%%%%%%%%%%%%%%%%%%%%%%%%%%%%%%

In \cite{Kwessi2015_7}, \cite{Kwessi2015_6}, and \cite{Kwessi2018_4}, hierarchical models for two and three species were discussed. In particular, the following model was discussed \begin{equation}\label{EQ:stocking}
	\begin{cases}
		x_{t+1} &=  H_1 x_t+x_t e^{\ds r_1 - x_t - \frac{m_1}{1+s_1x_t}} \vspace{0.125cm}\\
        y_{t+1} &=H_2 y_t+ y_t e^{\ds r_2 - y_t -bx_t - \frac{m_2}{1+s_2y_t}} 
	\end{cases}\;, 
\end{equation}
In this model, $x_t$ and $y_t$ represent the  densities of the two species $x$ and $y$ under consideration with $x$ being the ``stronger" species and $y$ being the weaker one. Moreover,  the $H_i$'s  are the immigration constants,  the $r_i$'s  are the species' growth rates, the $m_i$'s are the species'  mortality rates,  the $s_i$'s are the Allee effect constants, and $b$ is a competition constant. For $i=1,2$,   the  following assumptions were proposed to guarantee a strong Allee effect on each species:
\begin{itemize}
\item[(i)] $0<H_i<1, ~~r_is_i>1+s_i\ln(1-H_i)$\;, and $m_i>r_i-\ln(1-H_i)$\;.
\item[(ii)] $[r_is_i-1-s_i\ln(1-H_i)]^2>4s_i(m_i-r_i+\ln(1-H_i))$\;.
\end{itemize}
 Each species is subject to a strong  Allee effect induced multiplicatively with the terms $\ds e^{\frac{m_1}{1+s_1x_t}}$ and  $\ds e^{\frac{m_2}{1+s_2y_t}}$.  In this manuscript, we propose to discuss the following model with immigration, and Allee effect on the waker species, and a  Hollins type II functional response. The choice of a type II rather than a Type III stems from the fact that type III tends to occurs in experimental data but is rare in nature. 
 \begin{equation}
 \begin{cases} 
 \ds x_{t+1}=x_te^{\ds r_1-x_t+dh(x,y)+H_1}  \\
 y_{t+1}=y_te^{\ds r_2-y_t-\frac{m}{a+y_t}-h(x,y)+H_2}
 \end{cases}\;.
 \end{equation}
 with a Holling type II functional response $\ds h(x,y)=\frac{(b+cx_t)y_t}{1+py_t(b+cx_t)}$ and 
 
 \[
 \begin{array}{||l||l||l||}
\hline
 \mbox{Parameters} &\mbox{Denomination} &\mbox{Dimension}\\ \hline
 r_1,r_2 & \mbox{Populations growth rates} &(\mbox{Time})^{-1}\\ \hline
 H_1, H_2 & \mbox{Immigration/emigration rates} &(\mbox{Time})^{-1}\\ \hline
 a & \mbox{Allee Effect constant} & \mbox{Biomass}\\ \hline
b &\mbox{Species $x$ attack rate} & (\mbox{Biomass})^{-1}(\mbox{Time})^{-1} \\ \hline
 c & \mbox{Cooperation intensity between the two species}&(\mbox{Biomass})^{-2}(\mbox{Time})^{-1} \\ \hline
 d & \mbox{Conversion coefficient} & \mbox{Dimensionless}\\ \hline
  m & \mbox{Mortality rate  of species $y$ due to the Allee Effect}& \mbox{Time}^{-1} \\ \hline
 p & \mbox{Species $x$  handling time} & \mbox{Time} \\ \hline
 \end{array}
 \]
  \begin{remark}
One could also consider   a spatial or cluster model  for which immigration is dependent on the  distance between the center position $\xi$ of the cluster and the center $\eta$  from  which the species is immigrating from,  using a  Laplacian spatial dispersion kernel:
  \[
  K(\xi,\eta,\gamma)=\frac{1}{\gamma}e^{-\frac{\abs{\xi-\eta}}{\gamma}}\;.
\]
  For a given cluster/patch centered a $\eta$,  we will assume that the ecosystem has $N$ clusters each centered at $\xi_i$, for $j=1,2,\cdots, N$ and  individuals from each species $x$ and $y$ from these clusters move into the cluster centered at $\eta$ according to spatial dispersion kernels $K(\eta, \xi, \gamma)$. We will then have  the model.

 \begin{equation}
 \begin{cases} \ds x_{t+1}(\eta)=x_t(\eta)\sum_{j=1}^NK(\eta,\xi_j,\gamma_1)e^{\ds r_1-x_t(\xi_j)+dh(x_t(\xi_j),y_t(\xi_j))}  \\
\ds  y_{t+1}(\eta)=y_t(\eta)\sum_{j=1}^N K_2(\eta,\xi_j,\gamma_2)e^{\ds r_2-y_t(\xi_j)-\frac{m}{a+y_t(\xi_j)}-h(x_t(\xi_j),y_t(\xi_j))}
 \end{cases}\;.
\end{equation}
In the presence of multiple preys, say $M$, each with density $y_t^i$, ($i=1, \cdots M$) for the predator with density $x_t$  at time $t$, we could consider the general model
\begin{equation}
 \begin{cases} \ds x_{t+1}(\eta)=x_t(\eta)\sum_{j=1}^NK(\eta,\xi_j,\gamma_1)e^{\ds r_1-x_t(\xi_j)+\sum_{i=1}^Md_ih(x_t(\xi_j),y_t^i(\xi_j))}  & \\
\ds  y_{t+1}^i(\eta)=y_t^j(\eta)\sum_{j=1}^N K_2(\eta,\xi_j,\gamma_i)e^{\ds r_{i+1}-y_t^i(\xi_j)-\frac{m_i}{a_i+y_t^i(\xi_j)}-h(x_t(\xi_j),y_t^i(\xi_j))} & 
 \end{cases}\;.
\end{equation}
\end{remark}
 \subsection{Stability Analysis}
 \subsection{Fixed Points}
Recall that $r_1,r_2,a,b,c,d,p,m$ are all {\bf positive parameters}. \\
We start by finding the fixed points of the model; they are the origin $(0,0)$ and the intersections of the  isoclines of  equations 
\begin{eqnarray}
\ds r_1+H_1-x+\frac{d(b+cx)y}{1+py(b+cx)}&=& 0\;. \vspace{0.2cm}\\
\ds r_2+H_2-y-\frac{m}{a+y}-\frac{(b+cx)y}{1+py(b+cx)}&=& 0\;.
\end{eqnarray}
After simplifications, we will have 
\begin{eqnarray}
(C_1): \quad x&=&r_1+H_1+d(r_2+H_2)-dy-\frac{dm}{a+y}\;.\\
(C_2): \quad y &=& \frac{x-(r_1+H_1)}{(b+cx)(-px+d+p(r_1+H_1))}\;.
\end{eqnarray}
The fixed points will be the intersection between the curve ($C_2$) in the $xy$-plane with the $(C_1)$ in the $yx$-plane. Things to note:
\begin{itemize}
\item $(C_1)$ has a vertical asymptote $y=-a$ and an asymptote $x=r_1+H_1+d(r_2+H_2)-dy$ is the $yx$-plane.
\item $(C_2)$ has two vertical asymptotes $\ds x=-\frac{b}{c}$ and $x=\frac{d+p(r_1+H_1)}{p}$ and a horizontal asymptote $y=0$ in the $xy$plane.
\end{itemize}
%\begin{figure}[H]
%\centering\includegraphics[scale=0.4]{FixedPointsPlot.pdf}
%\caption{Fixed points of the model as intersection between the isoclines $(C_1)$ and $(C_2)$.}
%\label{fig:isoclines}
%\end{figure}
In all, we can expect the following fixed points.
\begin{itemize}
\item The origin $E_{00}(0,0)$.
\item Axial fixed points: 
\begin{enumerate}
\item A predator-free fixed point $E_{10}(0,y)$,
\item A  prey-free fixed $E_{01}(x,0)$.
\end{enumerate}
\item At most two interior fixed points $E_{11}(\alpha_{1x},\alpha_{1y})$ and $E_{22}(\alpha_{2x},\alpha_{2y})$. 
\end{itemize}

\begin{figure}[H]
\begin{center}
\begin{tabular}{cc}
({\bf a})& ({\bf b})\\
\includegraphics[scale=0.3]{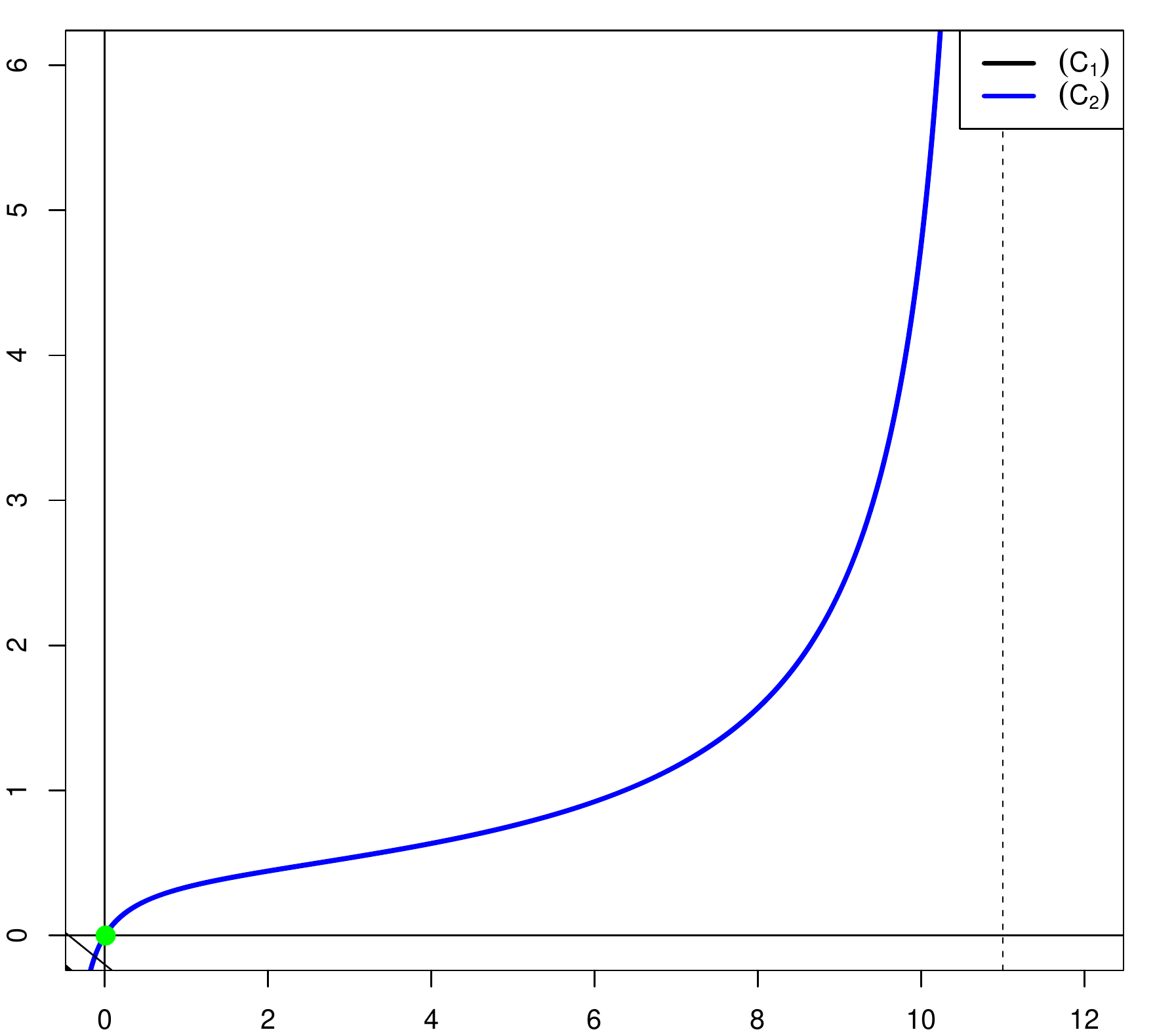} & \includegraphics[scale=0.3]{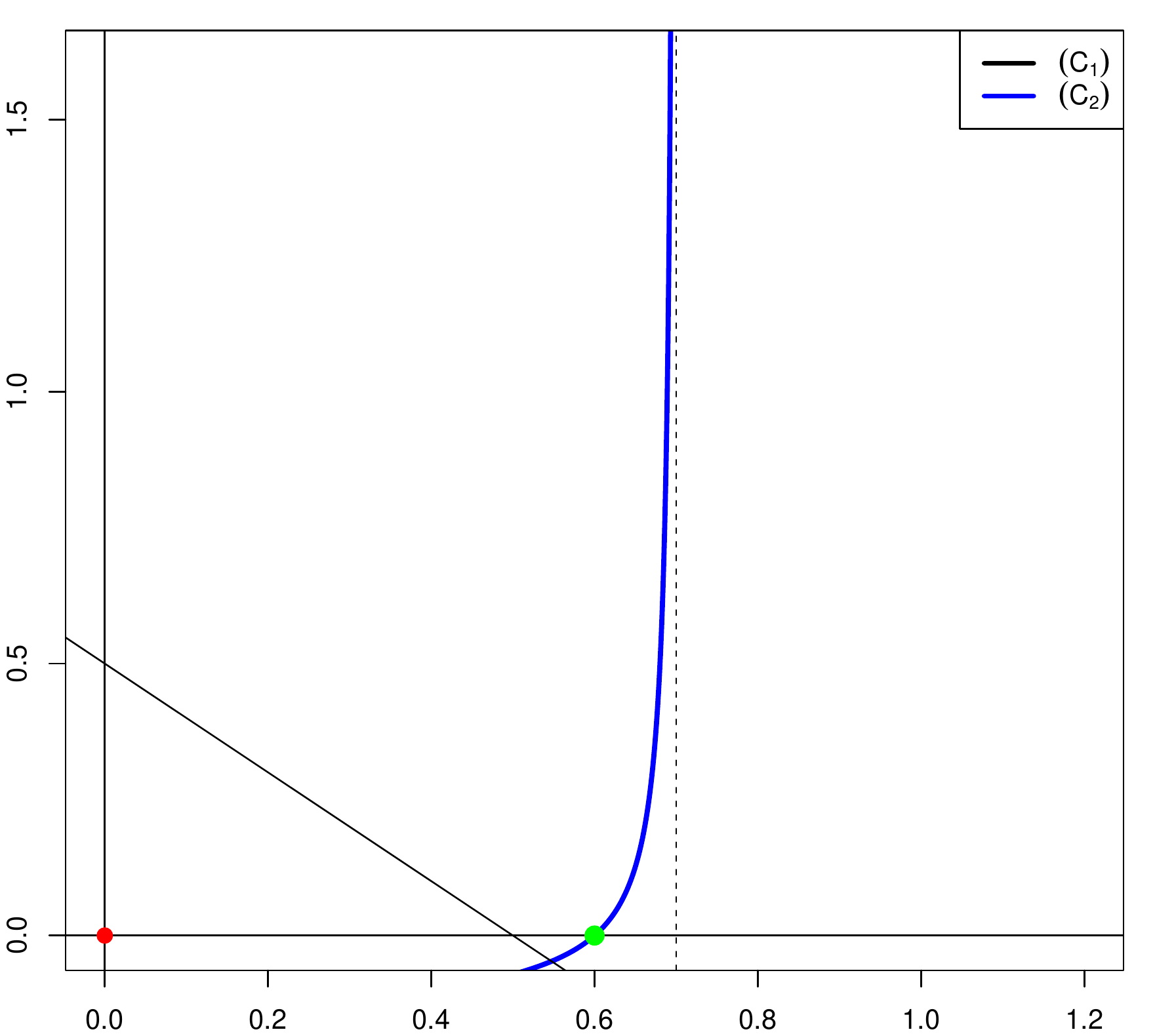}\\
\\
({\bf c}) & ({\bf d})\\
\includegraphics[scale=0.3]{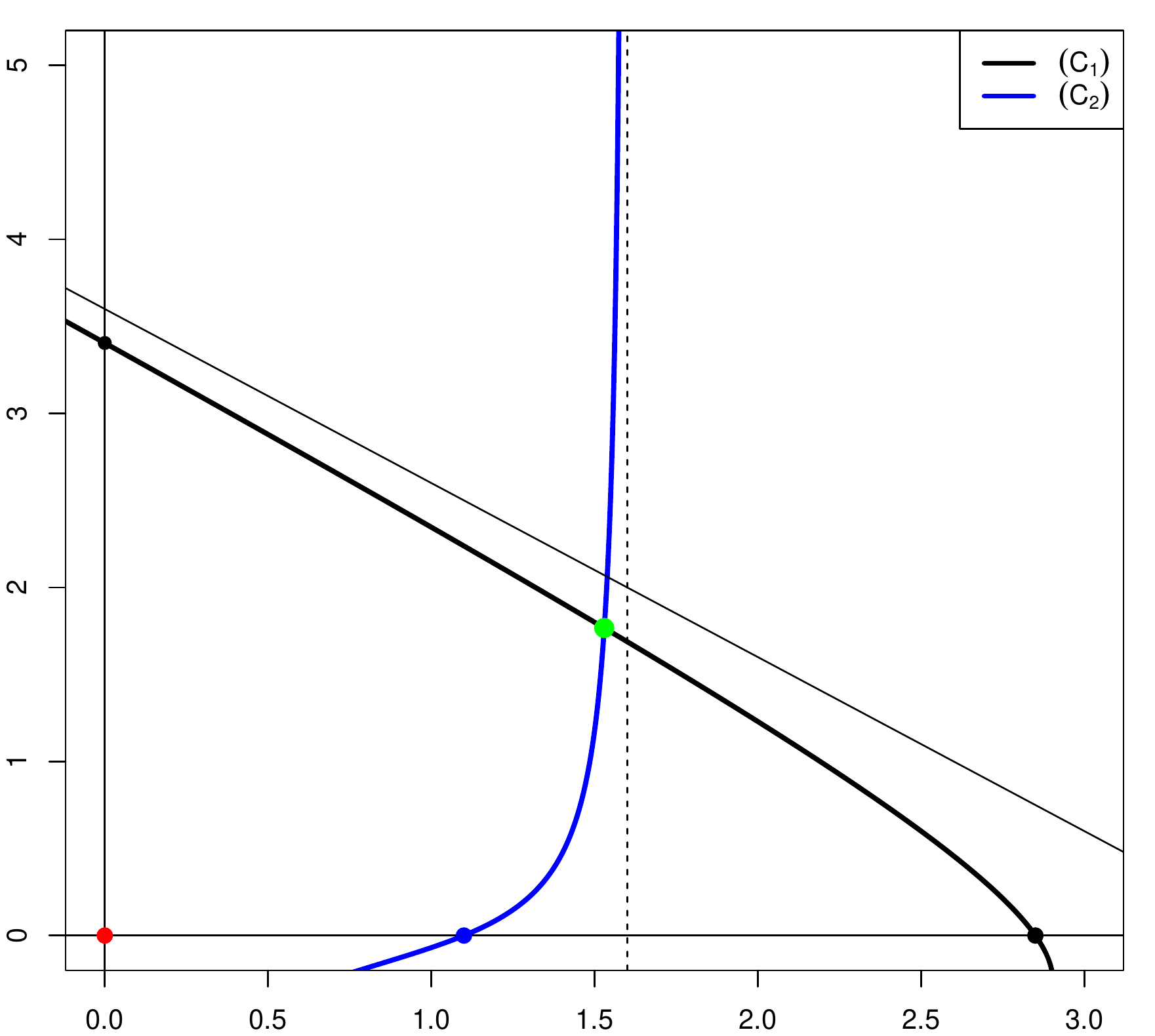} & \includegraphics[scale=0.3]{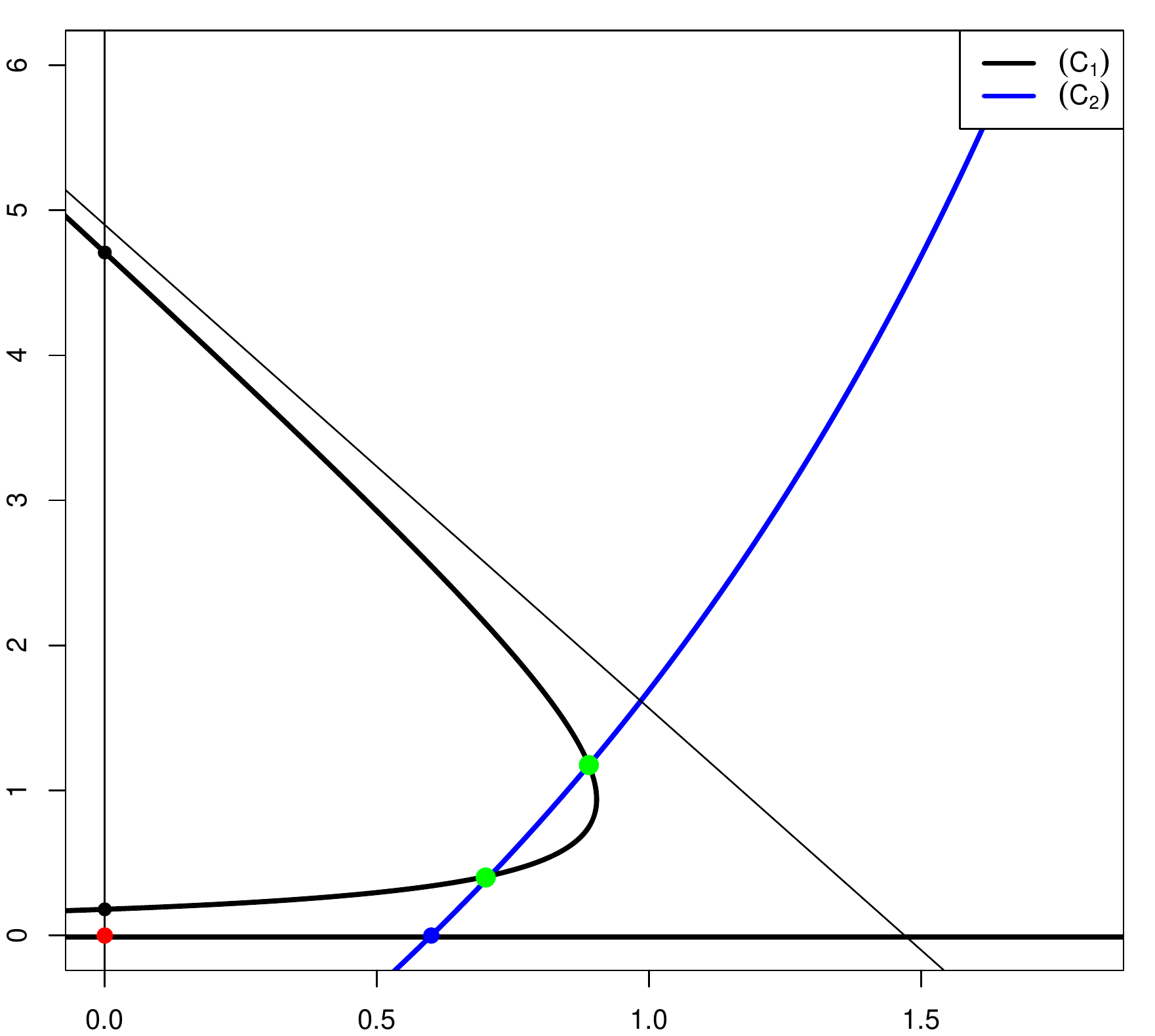}
\end{tabular}
\end{center}
\caption{Figure  ({\bf a}) represents the case of no interior fixed point, Figure  ({\bf b}) represents the case of one axial fixed point,  Figure ({\bf b}) represents the case of one interior  fixed point, and Figure  ({\bf d}) represents the case of two interior fixed points.}
\label{fig:AllFixedPoints}
\end{figure}
\subsubsection{Local Stability}
Let \begin{eqnarray*}
h(x,y)&=&\frac{(b+cx)y}{1+py(b+cx)}\;,\\
f_1(x,y)&=&e^{\ds r_1+H_1-x+dh(x,y)}\;,\\
f_2(x,y)&=& e^{\ds r_2+H_2-y-\frac{m}{a+y}-h(x,y)}\;,\\
F(x,y)&=&(F_1(x,y), F_2(x,y))=(xf_1(x,y), yf_2(x,y))\;.
\end{eqnarray*}
The Jacobian matrix of the system above at any given  point $M(x,y)$ is given as 
\[
J(x,y)=\begin{pmatrix}
f_1(x,y)\rb{1+x\rb{-1+\frac{dcy}{(1+py(b+cx))^2}}} & xf_1(x,y)\rb{\frac{d(b+cx)}{(1+py(b+cx))^2}}\\
yf_2(x,y)\rb{\frac{-cy^2}{(1+py(b+cx))^2}} & f_2(x,y)\left(1+y\left(-1+\frac{b+cx}{(1+py(b+cx))^2}\right)\right)
\end{pmatrix}\;.
\]
At the origin $E(0,0)$, the Jacobian is 
\[
J(0,0)=\begin{pmatrix}
f_1(0,0) & 0\\
0 & f_2(0,0)
\end{pmatrix}= \begin{pmatrix}
e^{r_1+H_1} & 0\\
0 & e^{\ds r_2+H_2-\frac{m}{a}}
\end{pmatrix}\;.
\]
Thus the eigenvalues are 
\[ \lambda_1=e^{r_1+H_1} \quad \mbox{and $\lambda_2=e^{r_2+H_2-\frac{m}{a}}$}\;.\]
At the predator-free axial fixed point $E(0,y)$ the Jacobian is 
\[
\begin{array}{ll}
J(0,y)&=\begin{pmatrix}
f_1(0,y) & 0\\
f_2(0,y)\rb{\ds \frac{-cy^3}{(1+pby))^2}} & f_2(0,y)\left(1+y\left(-1+\frac{b+}{(1+pby))^2}\right)\right)
\end{pmatrix}\mbox{}\\
& \mbox{}\\
&=\begin{pmatrix}
e^{r_1+H_1} & 0\\
e^{r_2+H_2-\frac{m}{a+y}}\rb{\ds \frac{-cy^3}{(1+pby))^2}} & e^{r_2+H_2-\frac{m}{a+y}}\left(1+y\left(-1+\frac{b}{(1+pby))^2}\right)\right)
\end{pmatrix}\;.
\end{array}
\]
Thus the eigenvalues are 
\[ \lambda_1=e^{r_1+H_1} \quad \mbox{and $\lambda_2=e^{r_2+H_2-\frac{m}{a+y}}\left(1+y\left(-1+\frac{b}{(1+pby))^2}\right)\right)$}\;.\]
At the prey-free axial fixed point $E(x,0)$ the Jacobian is
\[
\begin{array}{lll}
J(x,0)&= \begin{pmatrix}
f_1(x,0)(1-x) & xdf_1(x,0)(b+cx)\\
0 & f_2(x,0)
\end{pmatrix}
&=
\begin{pmatrix}
e^{r_1+H_1-x}(1-x) & xd(b+cx)e^{r_1+H_2-x}\\
0 & e^{r_2+H_2-\frac{m}{a}}
\end{pmatrix}\;.
\end{array}
\]
Thus the eigenvalues are \[ \lambda_1=e^{r_1+H_1-x}(1-x) \quad \mbox{and $\lambda_2=e^{r_2+H_2-\frac{m}{a}}$}\;.\]
At  an interior  fixed point  $E(x,y)$ with $x,y>0$, the Jacobian is 
\[
J(x,y)=\begin{pmatrix}
1+x\rb{-1+\frac{dcy}{(1+py(b+cx))^2}} & \frac{dx(b+cx)}{(1+py(b+cx))^2}\\
\frac{-cy^3}{(1+py(b+cx))^2} & 1+y\left(-1+\frac{b+cx}{(1+py(b+cx))^2}\right)
\end{pmatrix}\;,
\]
since $f_1(x,y)=f_2(x,y)=1$. Let $D$ be the determinant and $T$ be the trace of $J(x,y)$ respectively. Then the  eigenvalues are\\ 
\[ \ds \lambda_1=\frac{T-\sqrt{T^2-4D}}{2} \quad \mbox{and $\ds \lambda_2=\frac{T-\sqrt{T^2-4D}}{2}$}\;.\]
We the have the following result:
\begin{thm}
Consider the deterministic system above. Then 
\begin{itemize}
\item[$\blacktriangleright$] The origin $E_{00}(0,0)$ is locally asymptotically stable if and only if  $r_1+H_1<0$ and $r_2+H_2<\frac{m}{a}$.
\item[$\blacktriangleright$] The Predator-free equilibrium $E_{01}(0,y)$ is 
locally  asymptotically stable if and only if  $r_1+H_1<0$ and $\abs{1-y+\frac{yb}{(1+pby)^2}}<e^{\frac{m}{a}-r_2-H_2}$.
\item[$\blacktriangleright$] The Prey-free equilibrium  $E_{10}(0,y)$ is 
locally asymptotically stable if and only if  $\abs{x-1}<e^{x-r_1-H_1}$ and $r_2+H_2<\frac{m}{a}$.
 \item[$\blacktriangleright$] From the Determinant-Trace analysis, we know that 
\begin{itemize}
 \item If $\abs{\lambda_1}<1$ and  $\abs{\lambda_2}>1$ or $\abs{\lambda_1}>1$ and  $\abs{\lambda_2}<1$, then $E(x,y)$ is  locally a saddle point.
\item $\abs{\lambda_1}<1$ and  $\abs{\lambda_2}<1$, then $E(x,y)$ is  locally asymptotically stable.
\item $\abs{\lambda_1}>1$ and  $\abs{\lambda_2}>1$, then $E(x,y)$ is  locally  unstable.
\end{itemize}

\end{itemize}
\end{thm}

\subsection{Global Stability}
In this section, we discuss global stability condition for the origin and the interior equilibrium point.
\begin{thm} \label{thm:GlobalStat}
Put \[\mbox{$\alpha=e^{r_1+H_1+\frac{d}{p}}$ and $\beta=e^{r_2+H_2}$}\;.\]
\begin{itemize}
\item[$\blacktriangleright$] The origin $E_{00}(0,0)$ is globally asymptotically stable if $\alpha^2\leq 1$ and $\beta^2\leq 1$.
\item[$\blacktriangleright$] The interior equilibrium $E(x_*,y_*)$ is 
globally asymptotically stable if   $\alpha^2=1$ and  $\beta^2= 1$.

\end{itemize}
\end{thm}

%%%%%%%%%%%%%%%%%%%%%%%%%%%%%%%%%%%%%%%%%%%%%%%%%%%%%%%%%%%%%
\subsection{Simulations}
%%%%%%%%%%%%%%%%%%%%%%%%%%%%%%%%%%%%%%%%%%%%%%%%%%%%%%%%%%%%%

In this section, we illustrate the deterministic model  for the different types of equilibrium points. We chose eight starting points with respective $x$ and $y$ coordinates: $X_0=(0.1, 0.1, 0.1, 10, 10, 10, 2.5, 20), Y_0=(10, 0.5, 0.1, 0.1, 0.5,10, 0.1, 3)$. The trajectories are represented by the black, light green, light blue, light red, light cyan, orange, purple, and  magenta colors. The solid dots represent the fixed points of the model and the solid blue and light black curves represent the isoclines $(C_2)$ and $C_1$ respectively. 
\begin{figure}[H]
\begin{tabular}{ccc}
({\bf a}) & ({\bf b})  & ({\bf c}) \\
\includegraphics[scale=0.32]{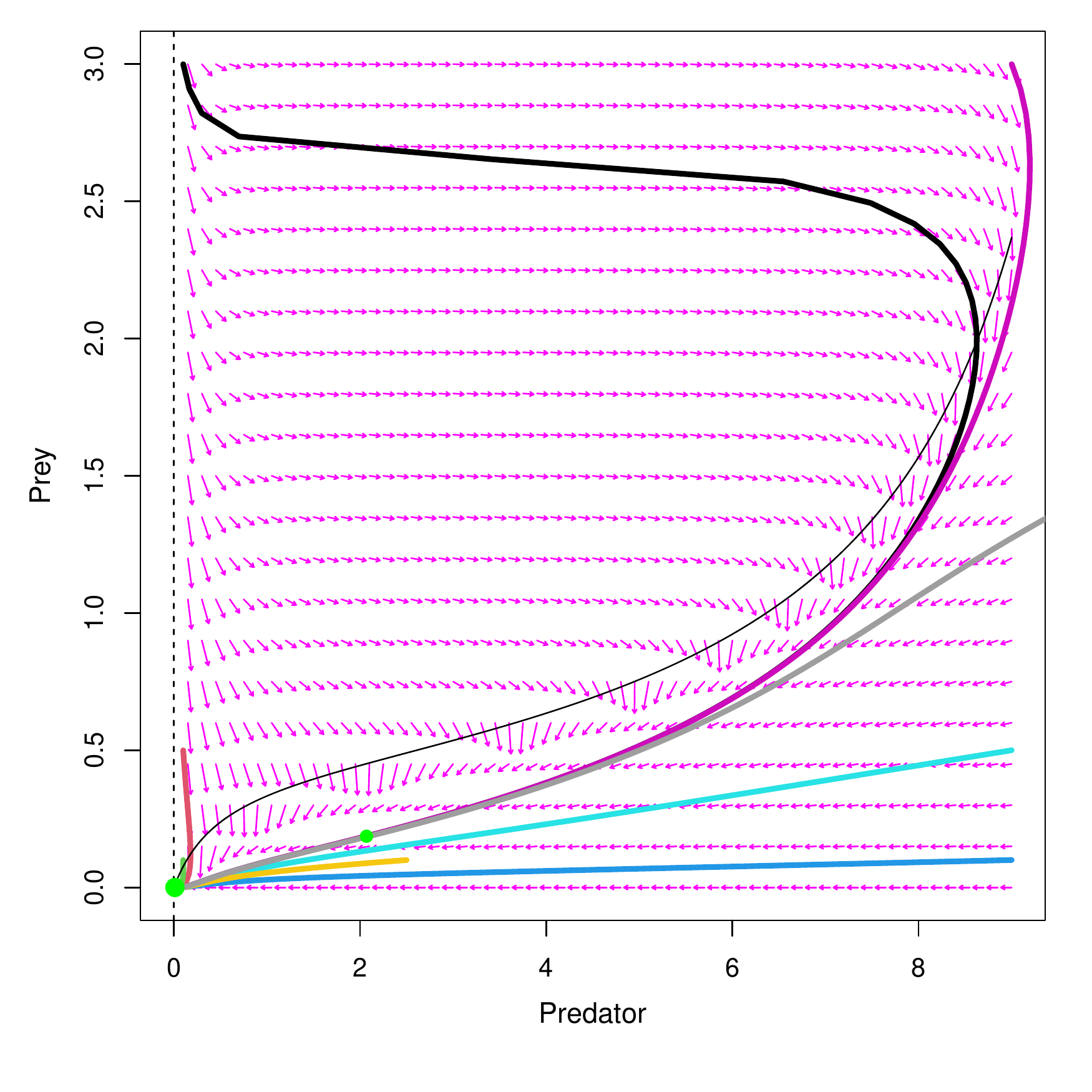} &  
\includegraphics[scale=0.32]{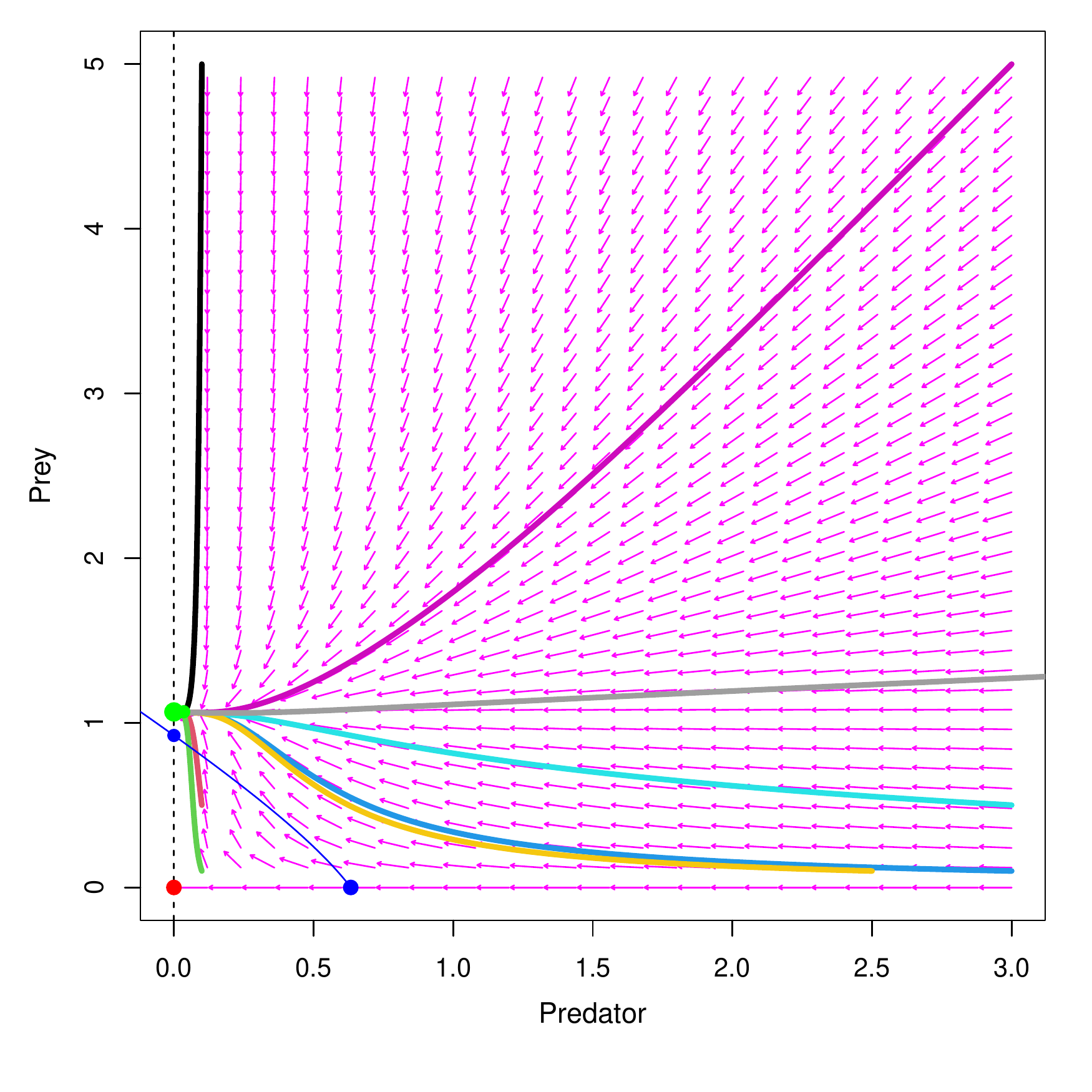}  &   \includegraphics[scale=0.32]{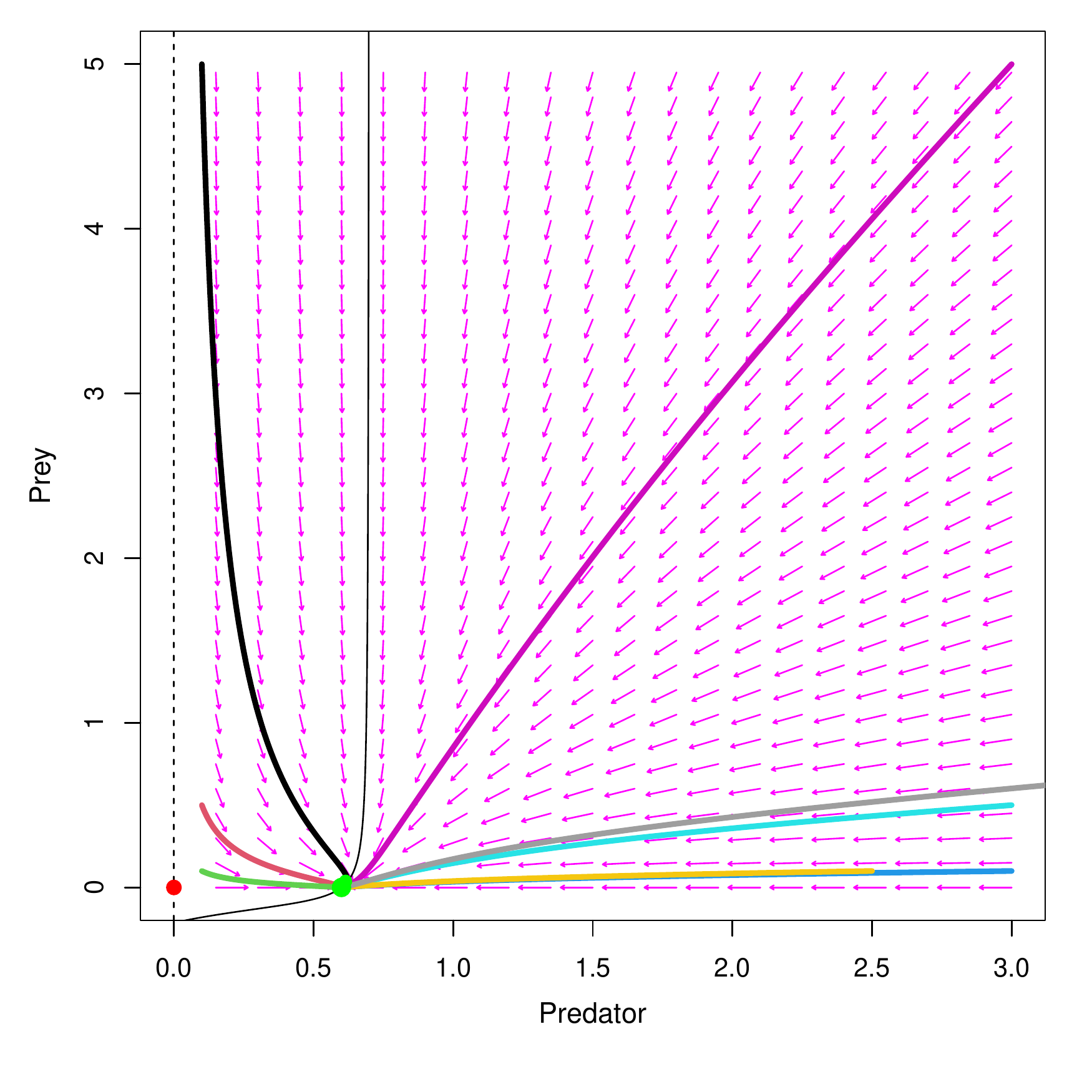} \\
({\bf d}) & ({\bf e}) & \\  
\includegraphics[scale=0.32]{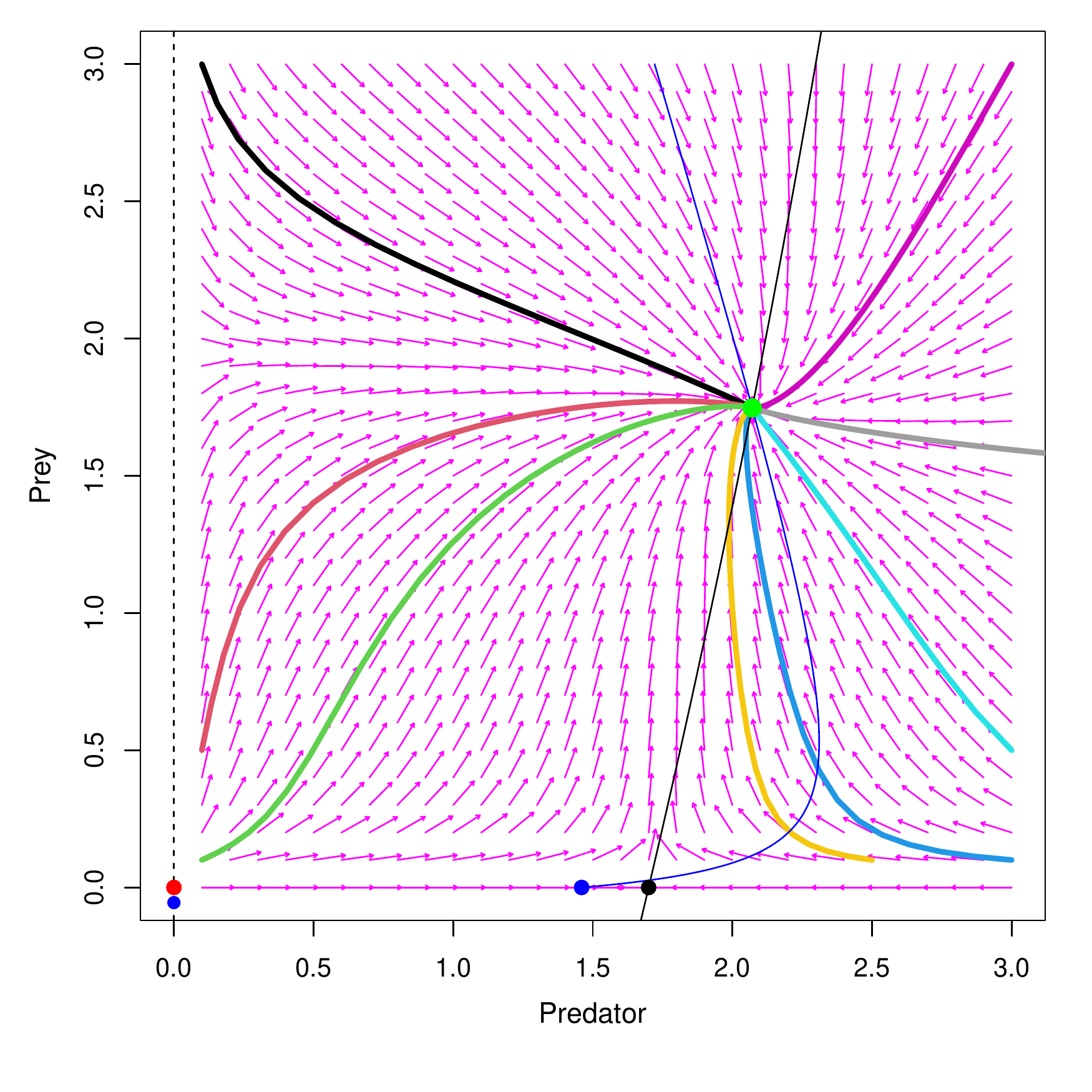} & \includegraphics[scale=0.32]{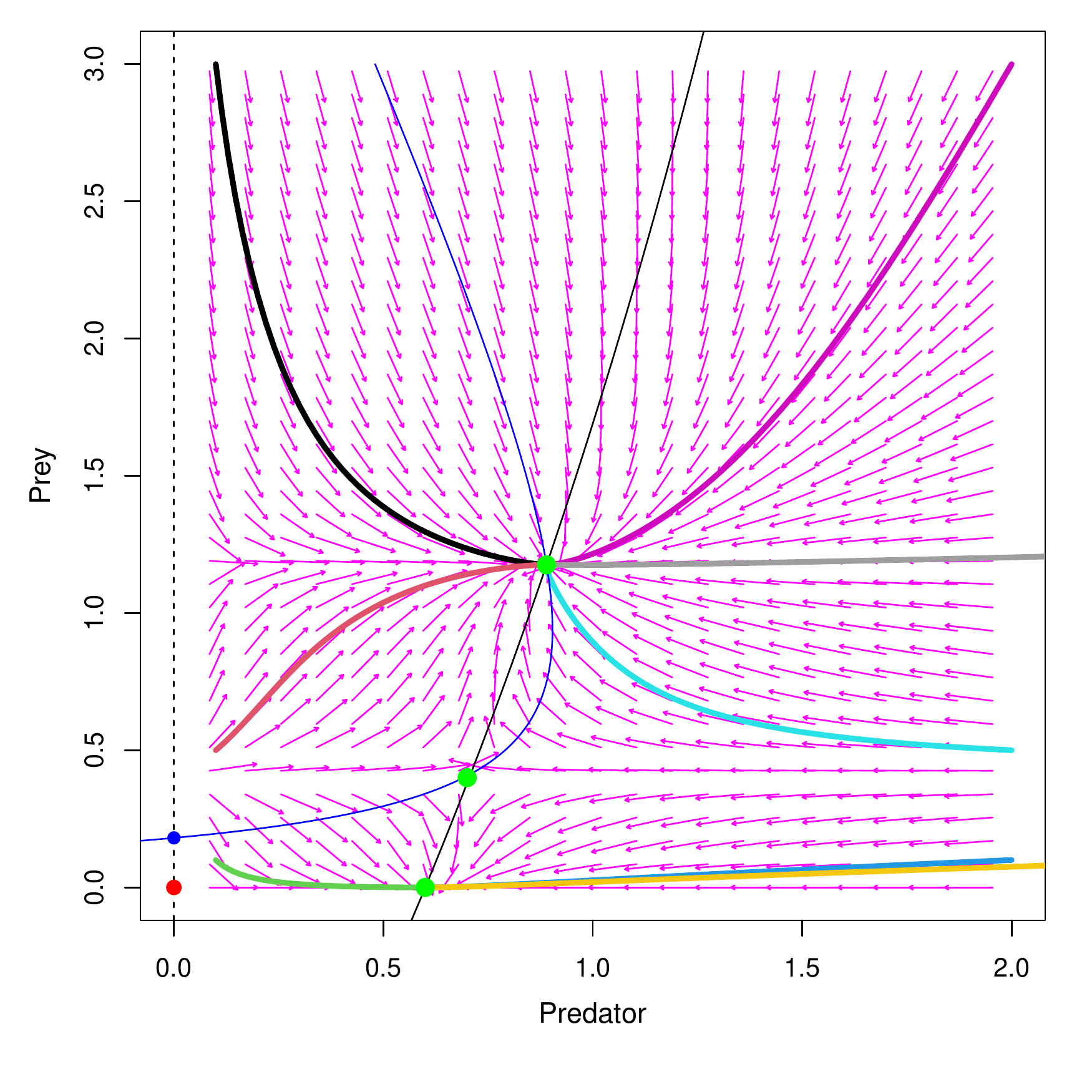}  & 
\end{tabular}
\caption{Figure ({\bf a}) show convergence of trajectories to the origin $E_00$, Figure ({\bf b}) shows convergence to the predator free axial point $E_{0y}$, Figure ({\bf c}) shows convergence to the prey free axial point $E_{x0}$, Figure ({\bf d}) shows convergence to the interior fixed point $E_{xy}$ Figure ({\bf e}) shows two interior fixed points, one stable and one unstable.}
\label{fig-Deterministic}
\end{figure}
%
%\begin{figure}[H]
%\begin{tabular}{lll}
%\includegraphics[scale=0.32]{PhaseSpaceNOIntEq.pdf} &  \includegraphics[scale=0.32]{TimesSeries1NOIntEq.pdf} &\includegraphics[scale=0.32]{TimesSeries2NOIntEq.pdf} 
%\end{tabular}
%\end{figure}
%\begin{figure}[H]
%\begin{tabular}{lll}
%\includegraphics[scale=0.32]{PhaseSpaceOneIntEq.pdf} &  \includegraphics[scale=0.32]{TimesSeries1OneIntEq.pdf} &\includegraphics[scale=0.32]{TimesSeries2OneIntEq.pdf}
%\end{tabular}
%\end{figure}
%
%\begin{figure}[H]
%\begin{tabular}{lll}
%\includegraphics[scale=0.32]{PhaseSpaceTwoIntEq.pdf} &  \includegraphics[scale=0.32]{TimesSeries1TwoIntEq.pdf} &\includegraphics[scale=0.32]{TimesSeries2TwoIntEq.pdf}
%\end{tabular}
%\end{figure}
%
%\begin{figure}[H]
%\begin{tabular}{lll}
%\includegraphics[scale=0.32]{PhaseSpaceOneYIntEq.pdf} &  \includegraphics[scale=0.32]{TimesSeries1OneYIntEq.pdf} &\includegraphics[scale=0.32]{TimesSeries2OneYIntEq.pdf}
%\end{tabular}
%\end{figure}
%
%\begin{figure}[H]
%\begin{tabular}{lll}
%\includegraphics[scale=0.32]{PhaseSpaceOneXIntEq.pdf} &  \includegraphics[scale=0.32]{TimesSeries1OneXIntEq.pdf} &\includegraphics[scale=0.32]{TimesSeries2OneXIntEq.pdf}
%\end{tabular}
%\end{figure}
%
%\centering\includegraphics[scale=0.7]{xAxialEquil.pdf}
%\centering\includegraphics[scale=0.7]{yAxialEquil.pdf}
%\centering\includegraphics[scale=0.7]{OneIntEquil.pdf}
%\centering\includegraphics[scale=0.7]{TwoIntEquil.pdf}

%%%%%%%%%%%%%%%%%%%%%%%%%%%%%%%%%%%%%%%%%%%%%%%%%%%%%%%%%%%%%
\subsection{Discussion}
%%%%%%%%%%%%%%%%%%%%%%%%%%%%%%%%%%%%%%%%%%%%%%%%%%%%%%%%%%%%%

We can make the following observation from the study above:\\
\noindent 1) The choice of the Holling functional certainly plays a role in this model. Our motivation for choosing type II rather than type III is that type III is found in population dynamics if the prey density is assumed constant, which is hardly the case here. \\
\noindent 2) Another observation is that our choice of type II functional response is different from the model suggested by Holling at inception. However, our functional response accounts for an important parameter in population dynamics such as the cooperation constant $c$.\\
\noindent 3) Our discussion of global stability is limited to using Lyapounov functional approach. It is important to point out that tools such monotone maps could also be used to prove global stability, see \cite{Balreira2014}.\\
\noindent 4) It  is known that in standard Ricker model with no immigration term, the intrinsic growth rate parameter is a bifurcation parameter (see for example \cite{Kwessi2018_3}), and in some case, the standard Ricker model possesses deterministic chaos. In the current  model, the bifurcations parameters will be $r_1+H_1+\frac{d}{p}$ and $r_2+H_2$ respectively for individual species and it highly likely that this system also possess deterministic chaos. Since this is beyond the investigation we are interested in pursuing, we will not discus it further for sake of brevity.

%%%%%%%%%%%%%%%%%%%%%%%%%%%%%%%%%%%%%%%%%%%%%%%%%%%%%%%%%%%%%
%%%%%%%%%%%%%%%%%%%%%%%%%%%%%%%%%%%%%%%%%%%%%%%%%%%%%%%%%%%%%
\section{Stochastic Model} \label{sect:stochatic}
%%%%%%%%%%%%%%%%%%%%%%%%%%%%%%%%%%%%%%%%%%%%%%%%%%%%%%%%%%%%%
%%%%%%%%%%%%%%%%%%%%%%%%%%%%%%%%%%%%%%%%%%%%%%%%%%%%%%%%%%%%%

To take into account environmental fluctuations on the species under consideration,  we let 
\begin{itemize}
\item $k\in \N$, we let $t_k=k\Delta t$ for some  $\Delta t>0$.
\item We define $X_k=(x_{_{t_k}},y_{_{t_k}})$ for $k\in \N$.
\item We let the initial condition be $X_0=(x_0,y_0)\in \R^+\times \R^+$.
\item We consider sequences of random variables $\ds \set{N_{j,k}(\Delta t),j=1,2,  k\in \N}$, normally distributed  such that for $j=1,2$ and for all $k\in \N$,
\begin{enumerate}
\item $E[N_{j,k}(\Delta t)]=0$.
\item $E[(N_{j,k}(\Delta t))^2]=\sigma_j^2 \Delta t$ with  for $j=1,2$.
\item $E[(N_{j,k}(\Delta t))^4]=o(\Delta t)$.
\end{enumerate}
\item We will assume that within the interval $[t_k, t_{k+1}]$, $X_k$ is affected by random perturbations $(x_{_{t_k}}N_{1,k}(\Delta t), y_{_{t_k}}N_{2,k}(\Delta t))$.
\end{itemize}
Therefore, for $k\in \N$, we will have 

\[\begin{cases}  x_{_{t_{k+1}}}-x_{t_k}=x_{_{t_k}}N_{1,k}(\Delta t)+x_{t_k}\rb{e^{\ds r_1-x_{t_k}+\frac{d(b+cx_{t_k})x_{t_k}}{1+py_{t_k}(b+cx_{t_k})}+H_1}}\Delta t  \\
 y_{_{t_{k+1}}}-y_{t_k}=y_{_{t_k}}N_{2,k}(\Delta t) +y_{_{t_k}}\rb{e^{\ds r_2-y_{_{t_k}}-\frac{m}{a+y_{_{t_k}}}-\frac{(b+cx_{_{t_k}})y_{_{t_k}}}{1+py_{_{t_k}}(b+cx_{_{t_k}})}+H_2}}\Delta t
 \end{cases}\;,\]
 or in vector form as 
 \[X_{k+1}-X_k=X_k F(X_k)\Delta t+X_kN_k(\Delta t)\;.\]
%If we rewrite as 
 %\[ X_{k+1}-X_k=X_kF(X_k)\Delta t+X_k(N_k(\Delta t)-1)\;,\]
 From It\^o's Calculus (see for instance \cite{Durrett1996}, if we put $W_{k+1,j}(\Delta t)-W_{k,j}(\Delta t):=N_{k,j}(\Delta t)$ and let   $\Delta t\to 0$, then  the equation above converges to an autonomous stochastic differential equation (SDE)
\begin{equation}\label{eqn:stochastic}
dX(t)=a(X(t))dt+b(X(t))dW(t)\;,
\end{equation}
  with initial condition $X_0=(x_0,y_0)$ and where 
 \[
 a(X(t))=\rb{\begin{array}{l}
 x(t)e^{\ds r_1-x(t)+\frac{d(b+cx(t))y(t)}{1+py(t)(b+cx(t))}+H_1}\\
 y(t)e^{\ds r_2-y(t)-\frac{m}{a+y(t)}-\frac{(b+cx(t))y(t)}{1+py(t)(b+cx(t))}+H_2}
 \end{array}}
 \;,\]
  and $b(X(t))=[x(t)\sigma_1,y(t)\sigma_2]$ represents the per-capita magnitude of environmental fluctuations,   and $W(t)=[W_1(t), W_2(t)]$ is a vector of Wiener processes.
  \begin{remark}
 It would be an important question to ask why the stochastic model  cannot be introduced by adding a ``stochastic" or ``random" term to the deterministic model or even to randomize the model parameters by assuming they are selected from specific probability distributions. While it would be a worthwhile effort theoretically,  the best answer to this question can be found in \cite{Hening2021}, Remark 2.1. Because of its importance, let us restate it here for self-containment: just adding a stochastic fluctuating term to a deterministic model has some short comings because it does not usually give a mechanism on  how different species are influenced by the environment. Instead, following the fundamental work by \cite{Turelli1977},  we see the SDE models as ``approximations for more realistic, but often analytically intractable, models". In particular, SDE's can be seen as scaling limits, or approximations, of  difference equations.
  \end{remark}
   We will be  concerned the existence of global solution, strongly persistence in mean, and the existence of a stationary distribution.
  
%%%%%%%%%%%%%%%%%%%%%%%%%%%%%%%%%%%%%%%%%%%%%%%%%%%%%%%%%%%%%
\subsection{Existence of global solution}
%%%%%%%%%%%%%%%%%%%%%%%%%%%%%%%%%%%%%%%%%%%%%%%%%%%%%%%%%%%%%

Let us start by recalling the following  Theorem on the existence of global solutions to a stochastic differential equation, see for instance \cite{Oksendal2014}, Theorem 5.2.1, p. 66.
\begin{thm}
Let $T>0$ and $A:\R^n\times [0,T]\to \R^n,\quad B:\R^n\times [0,T]\to \R^{n\times m}$ be measurable functions. Let  $X_0\in \R^n$ be a random variable such that $\mathbb{E}[\norm{X_0}^2]<\infty$, where $\norm{\cdot}$ is a norm in $\R^n$. Suppose that for given  
$t \in[0,T]$ and  $X(t), Y(t)\in \R^n$, there exist $K_1, K_2>0$ such that
\begin{enumerate}
\item  Linear growth condition:
\[\norm{A(X(t),t)}+\norm{B(X(t),t)}\leq K_1(1+\norm{X(t)})\;.\]
\item  Local Lipschitz Condition:
\[\norm{A(X(t),t)-A(Y(t),t)}+\norm{B(X(t),t)-B(Y(t),t)}\leq K_2\norm{X(t)-Y(t)}\;.\]
\end{enumerate}
Then the stochastic differential equation $dX(t)=A(X(t),t)+B(X(t),t)dW(t)$ with initial condition $X(0)=X_0$ has a unique solution $X(t)$ such that 
\[\mathbb{E}\Intv{\int_0^T \norm{X(t)}^2dt}<\infty\;.\]
\end{thm}

\begin{thm}
For any initial value  $X_0=(x_0,y_0)\in \R^+\times \R^+$, the above stochastic system has a unique positive global solution $X(t)=(x(t),y(t))\in \R^+\times \R^+$.
\end{thm}
\begin{proof}
For $n=2$, consider the above non-autonomous stochastic differential equation \[dX(t)=a(X(t))+b(X(t))dW(t)\;.\]
Put $X(0)=X_0=(x_0,y_0)$, for given $x_0,y_0\in \R^+$. Let $T>0$ be arbitrary.\\
Clearly, $\ds\mathbb{E}[\norm{X_0}^2]=\norm{X_0}^2<\infty$.
Recall that 
\[a(X(t))=[x(t)f_1(x(t),y(t)),y(t)f_2(x(t),y(t)];
 \quad b(X(t))=[x(t)\sigma_1,y(t)\sigma_2]\;.\]
 $\ds f_1(x(t),y(t))=e^{\ds r_1-x(t)+\frac{d(b+cx(t))y(t)}{1+py(t)(b+cx(t))}+H_1}\leq e^{r_1+H_1+\frac{d}{p}}$.\\
 $\ds f_2(x(t),y(t))=e^{\ds r_2-y(t)-\frac{m}{a+y(t)}-\frac{(b+cx(t))y(t)}{1+py(t)(b+cx(t))}+H_2}\leq e^{r_2+H_2}$.\\
 Using the Mean Value Theorem, we can choose \[K_1=K_2=\max\set{e^{r_1+H_1+\frac{d}{p}}+\sigma_1,e^{r_2+H_2}+\sigma_2}\;.\]
 Since this is true for any $T>0$, then we conclude that the solution exists and is global.

\end{proof}

%%%%%%%%%%%%%%%%%%%%%%%%%%%%%%%%%%%%%%%%%%%%%%%%%%%%%%%%%%%%%
\subsection{Strong Persistence in Mean}
%%%%%%%%%%%%%%%%%%%%%%%%%%%%%%%%%%%%%%%%%%%%%%%%%%%%%%%%%%%%%

Let us start by recalling the notion of strong persistence in mean and an important lemma used to prove that a stochastic differential equation is strongly persistent in mean.
\begin{defn}
Let $X(t)$ be the solution to a stochastic differential equation.
Suppose that for all $t>0$, the  normalized occupational measure or mean satisfies
\[\scp{X(t)}=\frac{1}{t}\int_0^tX(u)du<\infty\;.\]
$X(t)$ is said to be strongly persistent in mean if 
\[\lim_{t \to \infty} \inf \scp{X(t)}>0 \;.\]
\end{defn}
\noindent To prove that the solution of a stochastic differential equation is strongly persistent in mean, the following result is often used.
\begin{lemma} \label{lemMeanPers}
Let $W_i$ be  a Wiener Process for all $1\leq i\leq n$. If $X(t)\in (C[0,\infty], \R^+)$ and there exist positive constants $\gamma_1, \gamma_2$ and $T>0$ such that 
\[ \ds \ln(X(t))\geq \gamma_1 t-\gamma_{2} \int_0^tX(u)du+\sum_{i=1}^n\sigma_iW_i(t), \quad \mbox{\bf for $t\geq T$}\;,\]
then by the Strong Law of Large numbers, 
\[\lim_{t \to \infty} \inf\scp{X(t)}\geq \frac{\gamma_1}{\gamma_2}, \quad \mbox{almost surely}\;.\]
\end{lemma}
\noindent Now we can state the result on mean persistence pertaining to our system:
\begin{thm}\label{thm:MeanPersistence}
Let  $X(t)$ be the solution of the stochastic differential equation \eqref{eqn:stochastic}.
If \[\frac{1}{2}\max\set{\sigma_1^2e^{-r_1-H_1},\sigma_2^2e^{\frac{m}{a}+\frac{1}{p}-r_2-H_2}}<1\;,\] then $X(t)$ is strongly persistent in mean. 
\end{thm}
\begin{remark}  \label{rem:Lyapunov:exp}The above theorem can be proved differently. Indeed, the key of our proof  is showing that the quantity $\ds \lambda_i(\mu):=\int_0^{\infty}\rb{f_i(X(t))-\frac{\sigma_i^2}{2}}\mu(dX)>0$, where $\mu$ is the Lebesgue measure. This is referred to in 
\cite{Hening2022} as the {\it external Lyapunov exponent}, which determines the infinitesimal per-capita growth of species not supported by the measure $\mu$. In that paper, it is shown under certain conditions (see section \ref{sect:StatDist} below) that if $\lambda_i(\delta^*)>0$ (where $\delta^*$ is the Dirac measure concentrated at the origin),  then $X$ is strongly stochastically persistent. 
\end{remark}
%%%%%%%%%%%%%%%%%%%%%%%%%%%%%%%%%%%%%%%%%%%%%%%%%%%%%%%%%%%%%
\subsection{Existence of Stationary Distributions}\label{sect:StatDist}
%%%%%%%%%%%%%%%%%%%%%%%%%%%%%%%%%%%%%%%%%%%%%%%%%%%%%%%%%%%%%
 We are now concerned with the existence of a stationary distribution for our model. First,  let us define the notion of stationary distribution and strongly stochastic persistence along the lines of \cite{Hening2021}.
 \begin{defn} The probability measure $\xi$ is an invariant probability measure for a process (or a stochastic differential equation with solution) $X(t)$ if, whenever $X(0)=x$ has distribution $\xi$ , then for any time $t \geq 0$,  the distribution of $X(t)$ is given by $\xi$.
  \end{defn}
 \begin{defn}  Let  $\Sigma$ be a  $\sigma$-algebra  on $\R^{n,\circ}_+:=(0, \infty)^n$ and let  $E\in \Sigma$. We define the total variation norm as  \[\ds \norm{\mu (E)}_{TV}=\sup_{\pi}\set{\sum_{A \in \pi}\abs{\mu(A)}, \quad \mbox{$\pi$ is a countable disjoint partition of $E$}}\;.\]
 A process $X(t)$ with ${\bf x}=X(0)$ is said to be strongly stochastic persistent if it has a unique invariant probability measure $\xi_0$ defined on $\R^{n,\circ}_+$ and 
 \[\lim_{t\to \infty}\norm{\mathbb{P}(t,X,\cdot)-\xi_0(\cdot)}_{TV}=0, \quad \mbox{for ${\bf x}\in \R^{n,\circ}_+$}\;.\]
 
 \end{defn}
%Let us recall the following result, which can be found in various forms in \cite{Hening2021} or \cite{Hening2022}.
\noindent We can now state the result on stationary measures related to our model.

\begin{thm}\label{thm:StatDist}
Let  $X(t)$ be the solution of the stochastic differential equation \eqref{eqn:stochastic}.
If \[\frac{1}{2}\max\set{\sigma_1^2e^{-r_1-H_1},\sigma_2^2e^{\frac{m}{a}+\frac{1}{p}-r_2-H_2}}<1\;,\] 
then the stochastic differential equation above has a unique stationary distribution $\xi_0$ with support on $\R^{2,\circ}_+$ and  \[\lim_{t\to \infty}\norm{\mathbb{P}(t,X,\cdot)-\xi_0(\cdot)}_{TV}=0, \quad \mbox{for ${\bf x}\in \R^{2,\circ}_+$}\;.\]

\end{thm}
\begin{remark}
The proof relies of checking the key assumptions of Theorem 2.1 in \cite{Hening2021}. This theorem is quite strong in that it proves strong stochastic persistence, which is a stronger notion than persistence in probability and almost sure persistence in probability. Another advantage of Theorem 2.1 is that it is easy to extend it to multiple species.
Often, to show existence of a stationary distribution, one would  show almost sure persistence as in  the proposition below. However, it is not easier to check the second part of the proposition, which technically amounts to proving   Theorem 2.1 in \cite{Hening2021} from scratch. 
\end{remark}
\begin{proposition}
The Markov process $X(t)$ has  a unique stationary distribution $\xi$ if there is a bounded domain $V\in \R^n$ with a regular boundary $\Sigma$ for which the following are hold true:
\begin{enumerate}
\item In the domain $V$ and some of its neighborhood, the smallest eigenvalue of the associated diffusion matrix $\Gamma(x)$ is far from zero.
 \item If ${\bf x}\in \R^n\setminus V$, the mean time $\tau({\bf x})$ required for any path emerging from ${\bf x}$ reaching the set $V$ is finite and $\ds \sup_{{\bf x}\in U}\; \mathbb{E}[\tau({\bf x})]<\infty$ for every compact set $U\subset \R^n$, then for any integrable function $f$ with respect to the measure $\xi$, 
 \[\mathbb{P}\rb{\lim_{T\to \infty}\frac{1}{T}\int_0^T f(X(t))dt=\int_{\R^n}f(x)\xi(dx)}=1\quad \mbox{for all $x\in \R^n$}\;.\]
 %Or 
% \[\scp{f(X(T))}\overset{\mbox{a.s.}}{\longrightarrow} \int_{\R^2}f(x)\xi(dx)\quad \mbox{when $T\to \infty$}\;.\]
\end{enumerate}
\end{proposition}

\subsection{Simulations}
To simulate the behavior of the stochastic model on an interval $[0,T]$ for a given $T>0$, we will use the following algorithm due to \cite{Milstein1975}.
\begin{enumerate}
\item Select an initial condition $X_0=(x_0,y_0)$
\item Select and integer $N>1$.
\item Partition the interval $[0,T]$ as $0=t_0<t_1<t_2<\cdots<t_N=T$ with $t_k=k\Delta t$ with $ \ds \Delta t= \frac{T}{N}$.
\item Use Milstein Higher Order Scheme, that is, for $0\leq k\leq N$.
\[
\begin{cases}
x_{k+1}=&x_k+x_kf_1(x_k,y_k)\Delta t+x_k\sigma_1\Delta W_{1,k}\sqrt{\Delta t}+x_k\frac{\sigma_1^2}{2}[(\Delta W_{1,k})^2-\Delta t]\\
y_{k+1}=&y_k+y_kf_2(x_k,y_k)\Delta t+y_k\sigma_2\Delta W_{2,k}\sqrt{\Delta t}+y_k\frac{\sigma_2^2}{2}[(\Delta W_{2,k})^2-\Delta t]
\end{cases}\;.
\]
\noindent where $\Delta W_{i,k}=W_{k+1}-W_{k}$ are  independent  normally distributed random variables with zero mean and variance $\Delta t$.
\end{enumerate}
In all figures below, we consider $N=100, T=1$, and the model parameters corresponding to a two interior fixed points for the deterministic model: 
$r_1=0.5, r_2=0.1, H_1=1.2,  H_2=1.8,  m=0.4,  a=0.1,  b=0.9,  c=0.01,  d=0.3,  p=0.1$.\\
In the  first panel on the right, we plot the phase space diagram of the deterministic model  overlayed with 100 random trajectories   with five initial  points  with  coordinates  given by $X_0=(0.25, 0.25, 2.75, 2.75, 2.75), Y_0=(2, 0.25, 2, 0.25, 0.1)$. The deterministic trajectories are represented by the solid thick red curves.
The second  and fifth panels are histograms of stationary distributions of prey and predator respectively. The third and fourth panels represent the levels curves  and  three dimensional representation of the stationary distribution to have more perspective. In the third plots, the black lines represent respectively the estimated sample mean of densities of both predator and prey. What they show is how stochasticity shifts the stable fixed points of deterministic models.

%%%%%%%%%%%%%%%%%%%%%%%%%%%%%%%%%%%%%%%%%%%%%%%%%%%%%%%%%%%%%
\subsubsection{Low stochasticity  on both species}
%%%%%%%%%%%%%%%%%%%%%%%%%%%%%%%%%%%%%%%%%%%%%%%%%%%%%%%%%%%%%

  \begin{figure}[H]
  \centering  \resizebox*{5.9in}{!}{
 \begin{tabular}{ccc}
\includegraphics[scale=0.29]{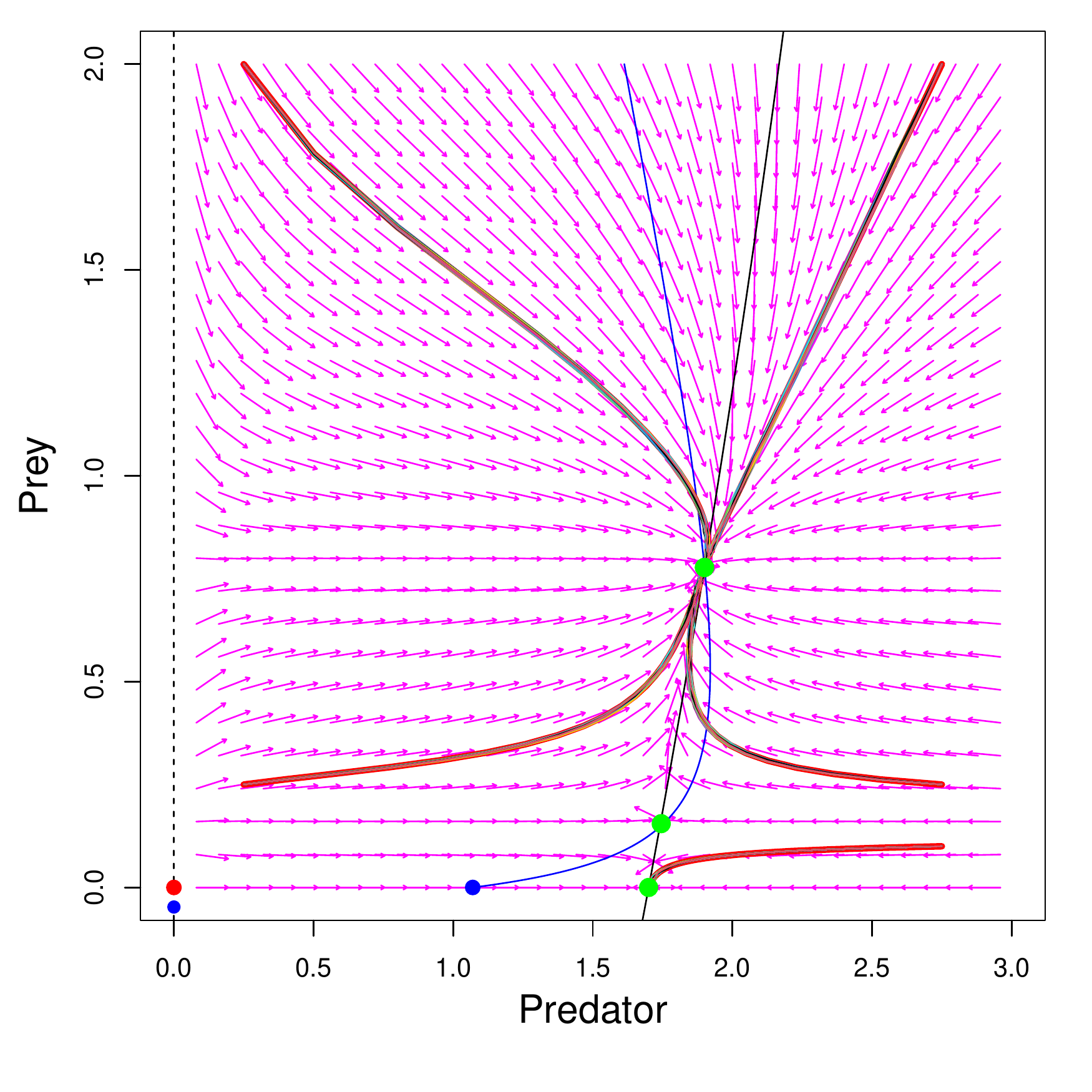} 
 &  \includegraphics[scale=0.3, angle =90]{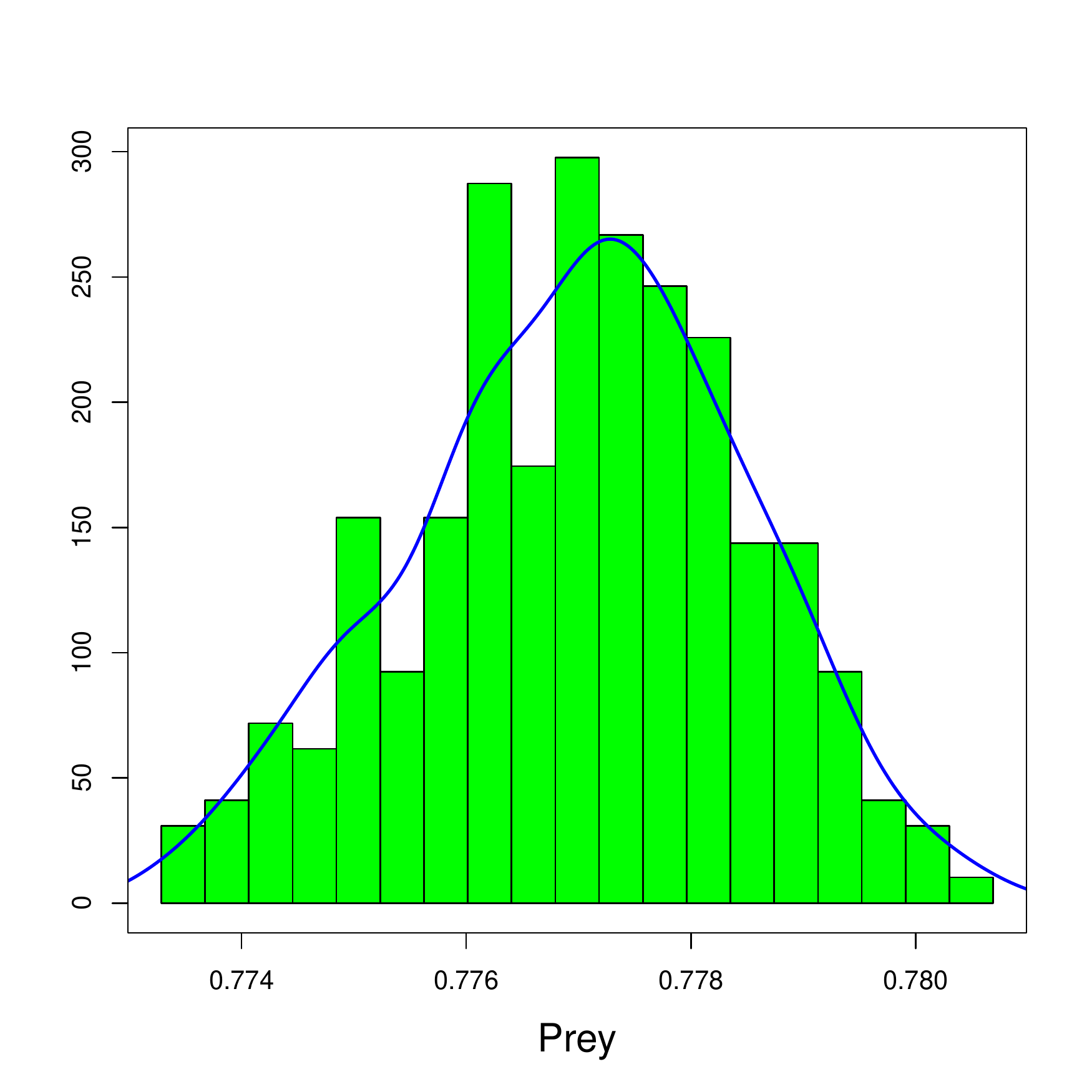} &   \includegraphics[width=2in, height=2.05in]{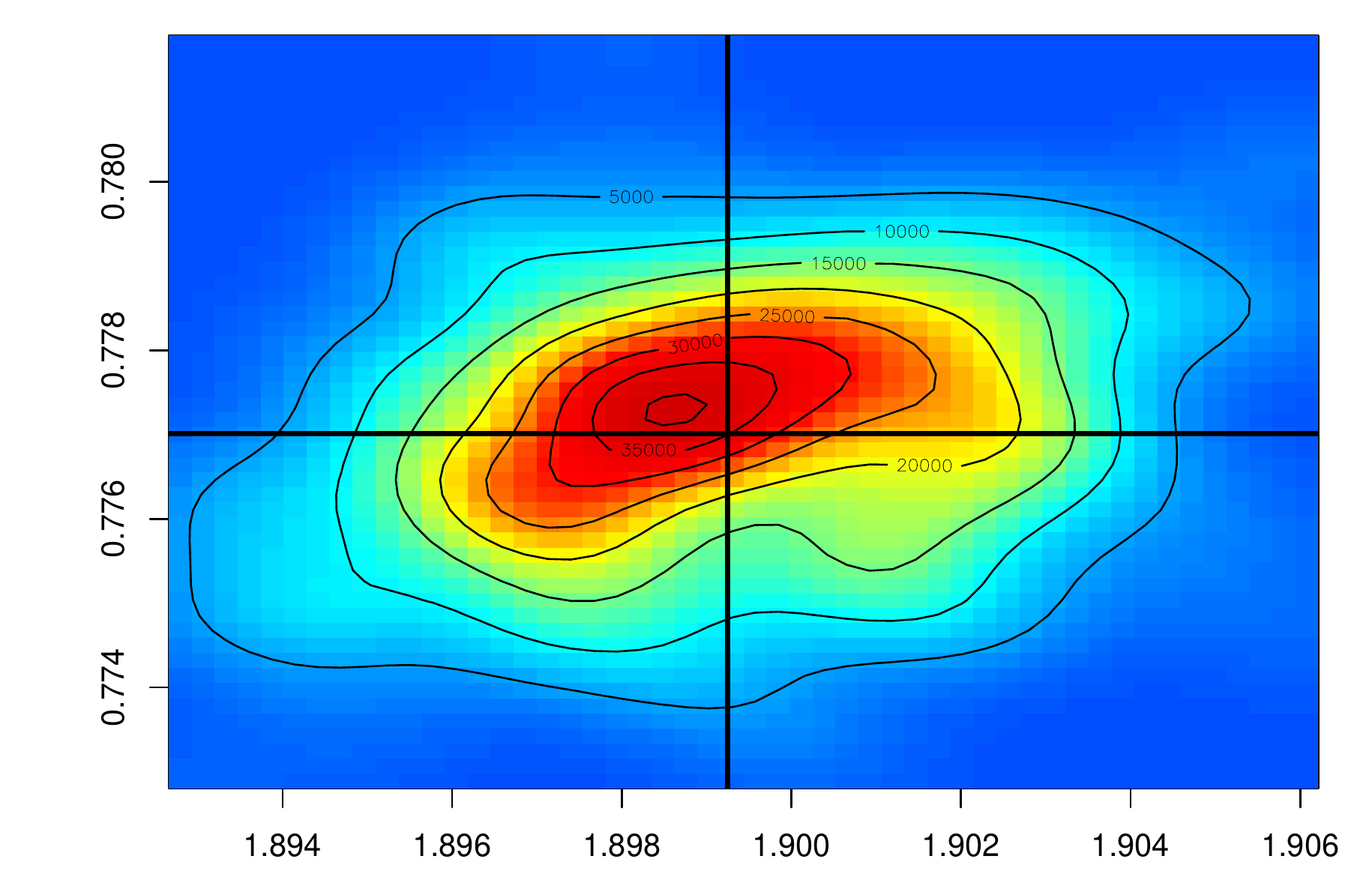} \\  
  &\includegraphics[scale=0.4]{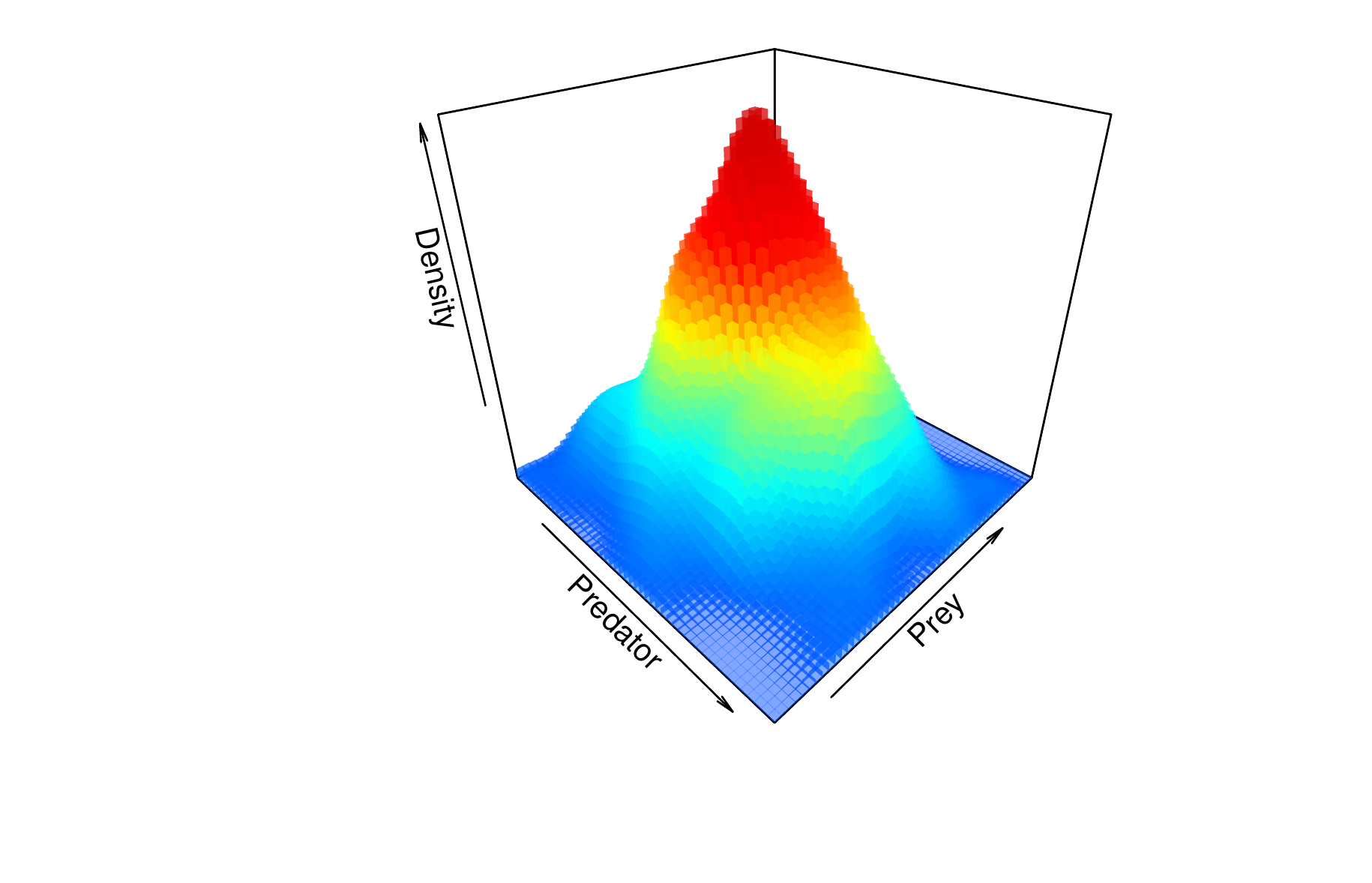} & \includegraphics[scale=0.32]{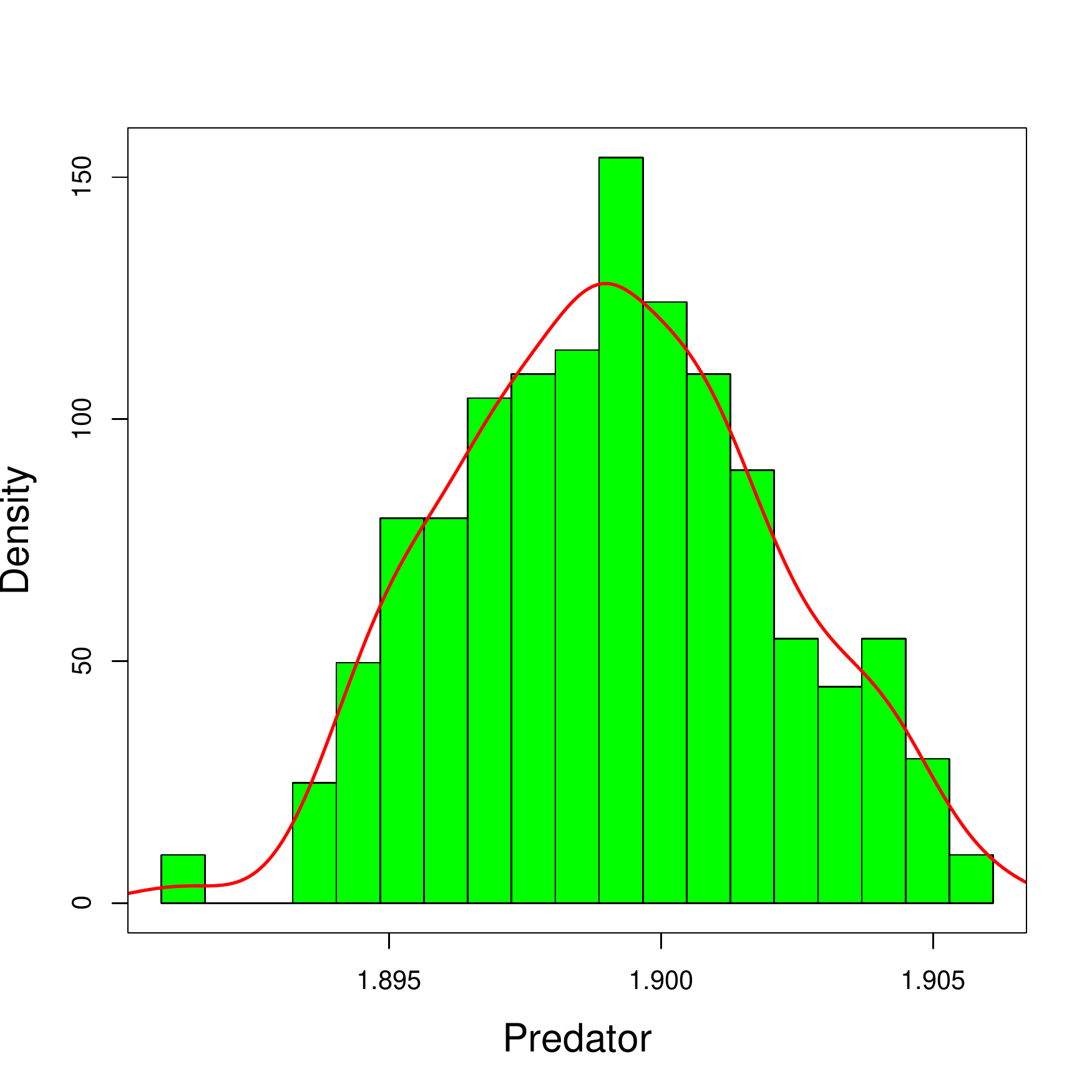} 
 \end{tabular}}
 \caption{In this case, we let $\sigma_1=\sigma_2=0.1$. Unsurprisingly, the trajectories are similar to that of the deterministic model, with densities close to symmetric.}
 \label{Fig}
 \end{figure}

%%%%%%%%%%%%%%%%%%%%%%%%%%%%%%%%%%%%%%%%%%%%%%%%%%%%%%%%%%%%%
\subsubsection{\bf Low stochasticity  on predator, high on prey}
%%%%%%%%%%%%%%%%%%%%%%%%%%%%%%%%%%%%%%%%%%%%%%%%%%%%%%%%%%%%%

 \begin{figure}[H]
  \centering  \resizebox*{5.9in}{!}{
 \begin{tabular}{ccc}
 \includegraphics[scale=0.29]{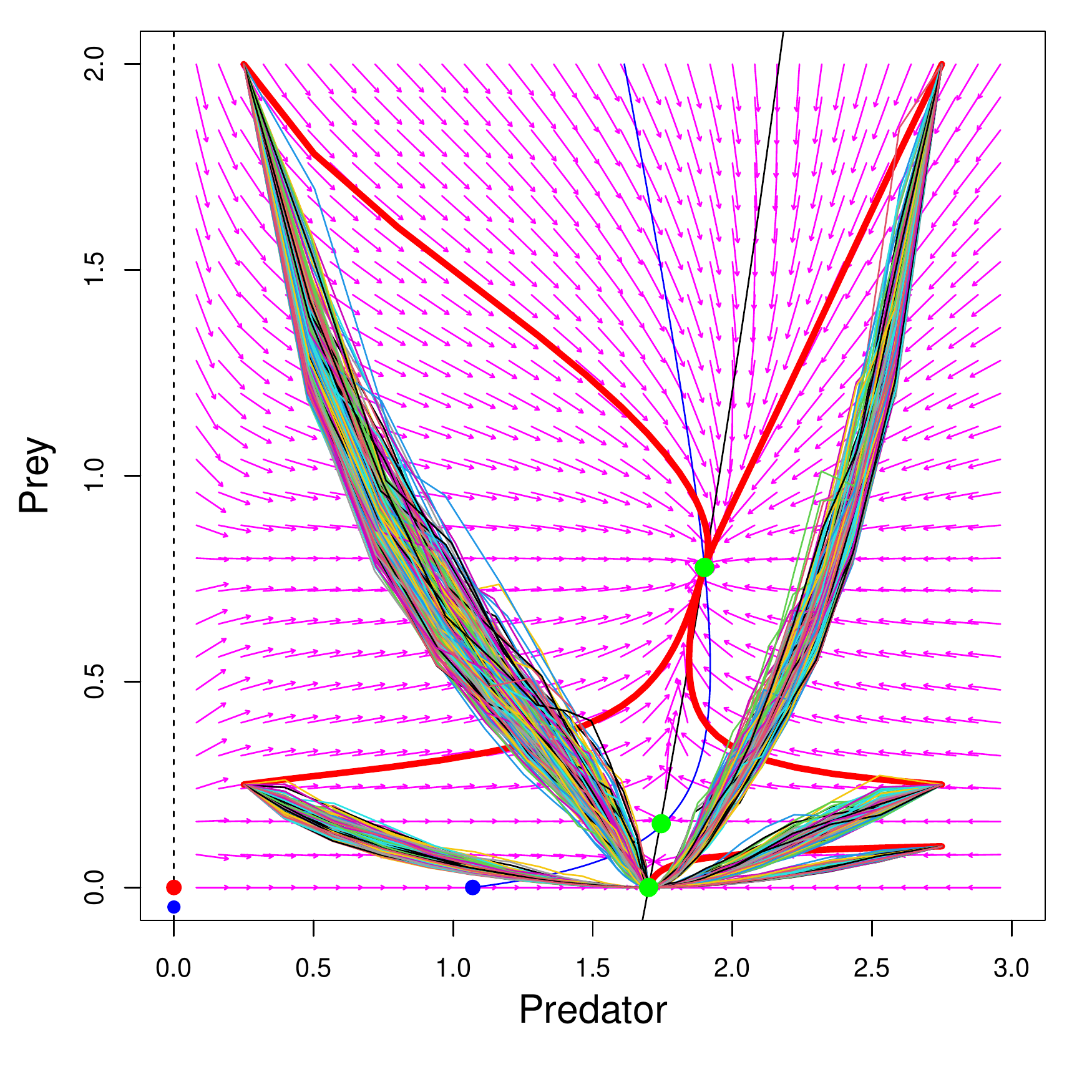} &  \includegraphics[scale=0.3, angle =90]{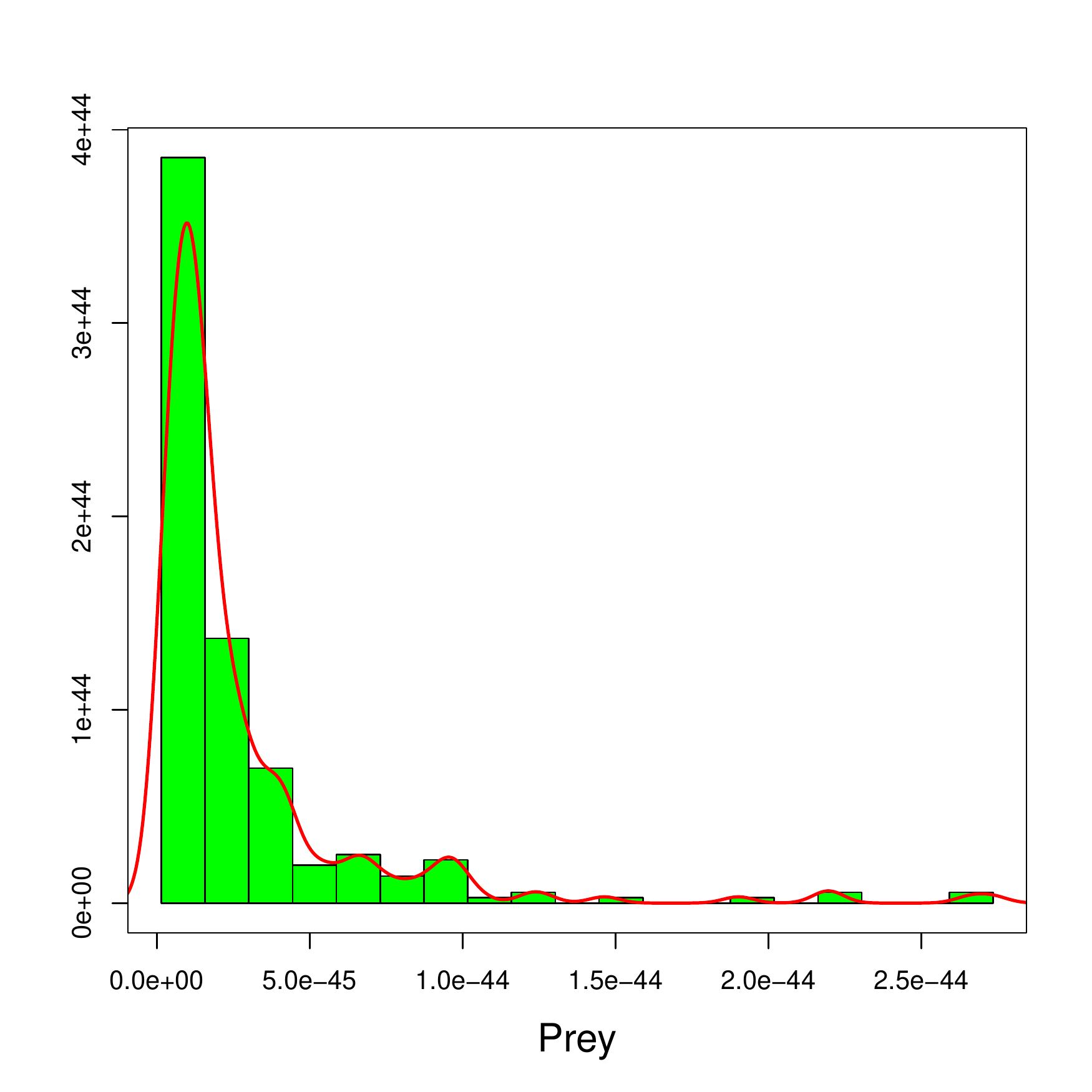} &   \includegraphics[width=2in, height=2.05in]{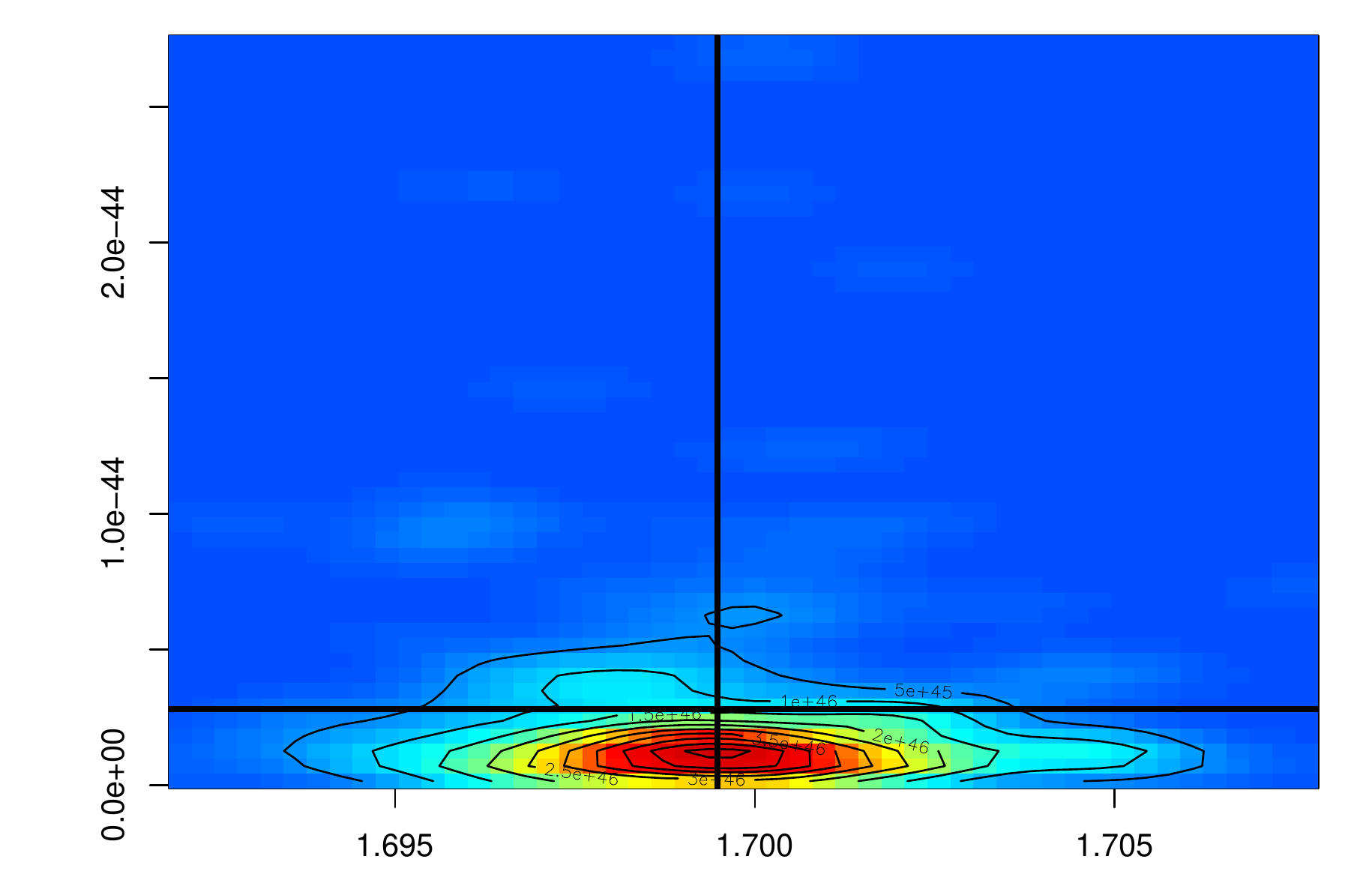} \\  
  &\includegraphics[scale=0.4]{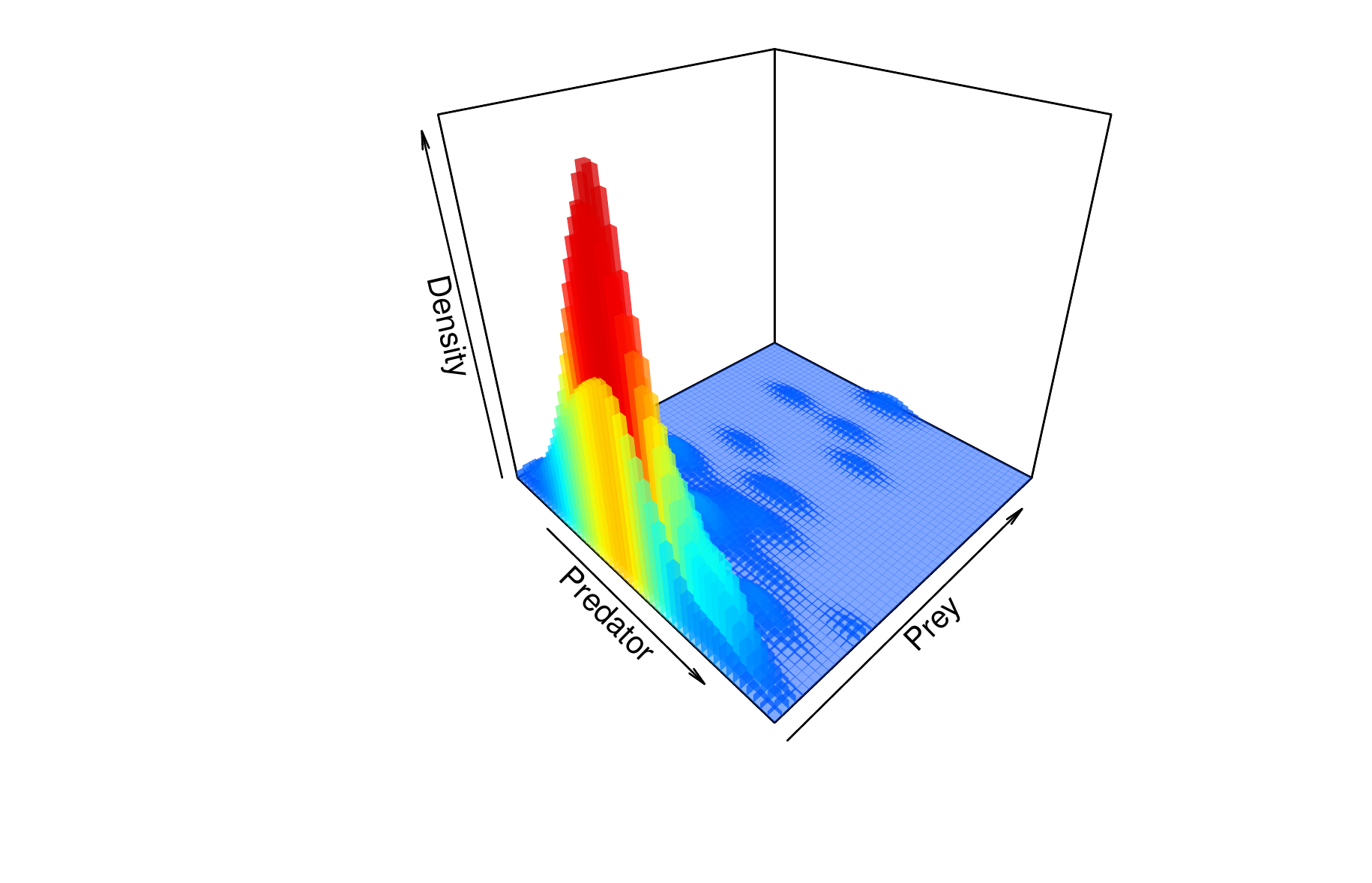} & \includegraphics[scale=0.32]{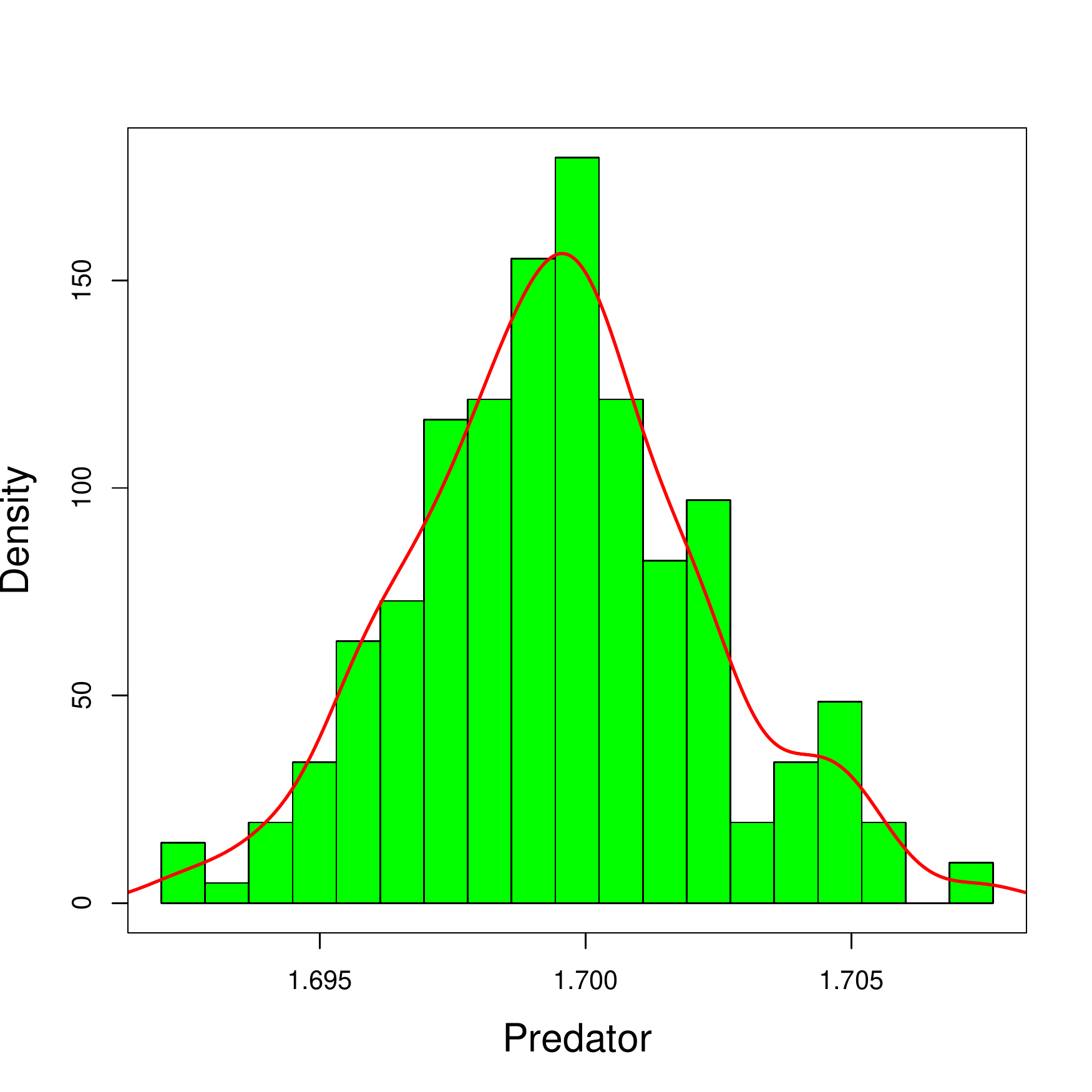} 
 \end{tabular}}
  \caption{In this case, we let $\sigma_1=6$ and $\sigma_2=0.1$. In this case, the predator survives and the prey quickly approaches limits of extinction.
}
 \label{Fig}
 \end{figure}
%%%%%%%%%%%%%%%%%%%%%%%%%%%%%%%%%%%%%%%%%%%%%%%%%%%%%%%%%%%%%
\subsubsection{\bf High stochasticity   on predator, low  on prey}
%%%%%%%%%%%%%%%%%%%%%%%%%%%%%%%%%%%%%%%%%%%%%%%%%%%%%%%%%%%%%

 \begin{figure}[H]
\centering \resizebox*{5.9in}{!}{
 \begin{tabular}{ccc}
 \includegraphics[scale=0.29]{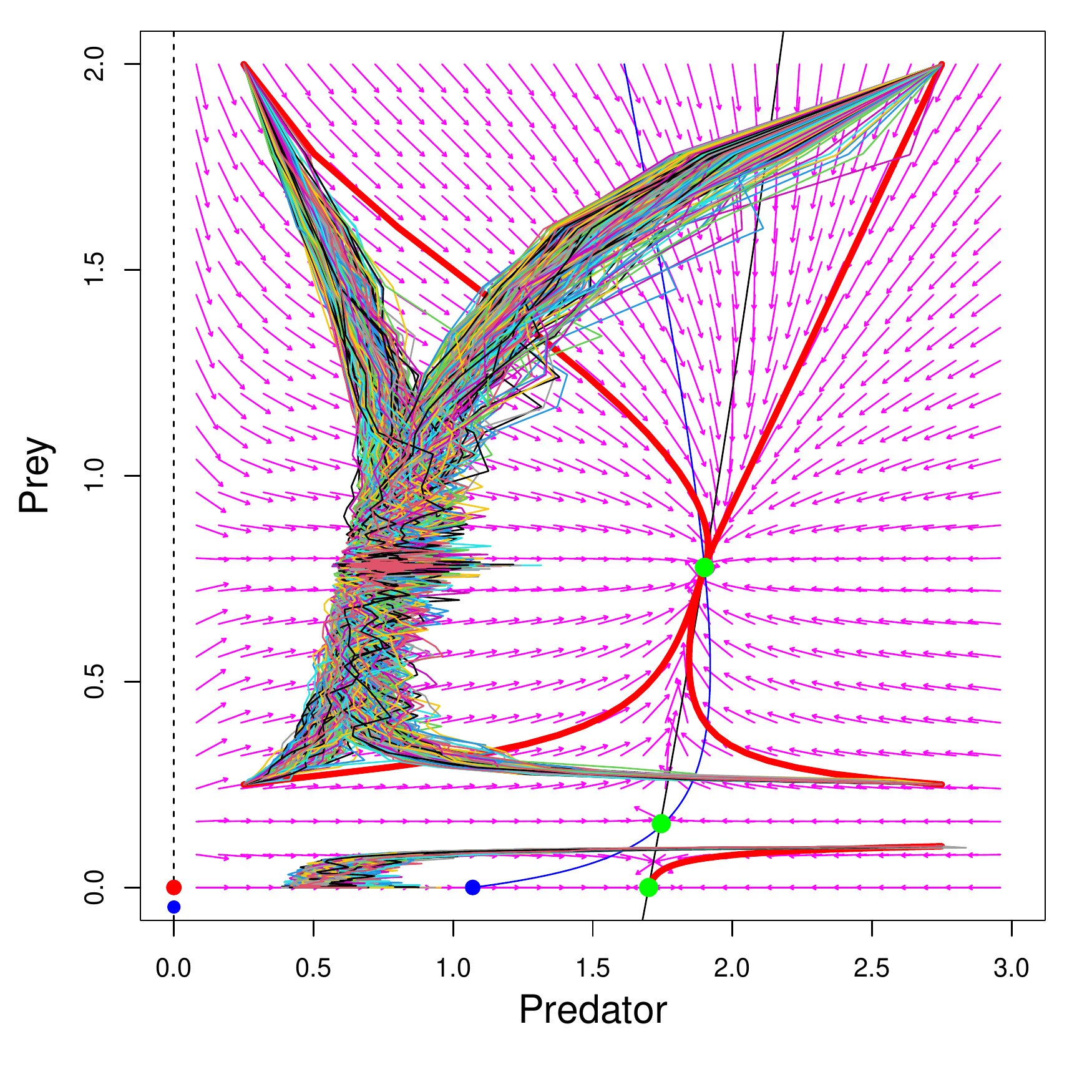} &  \includegraphics[scale=0.3, angle =90]{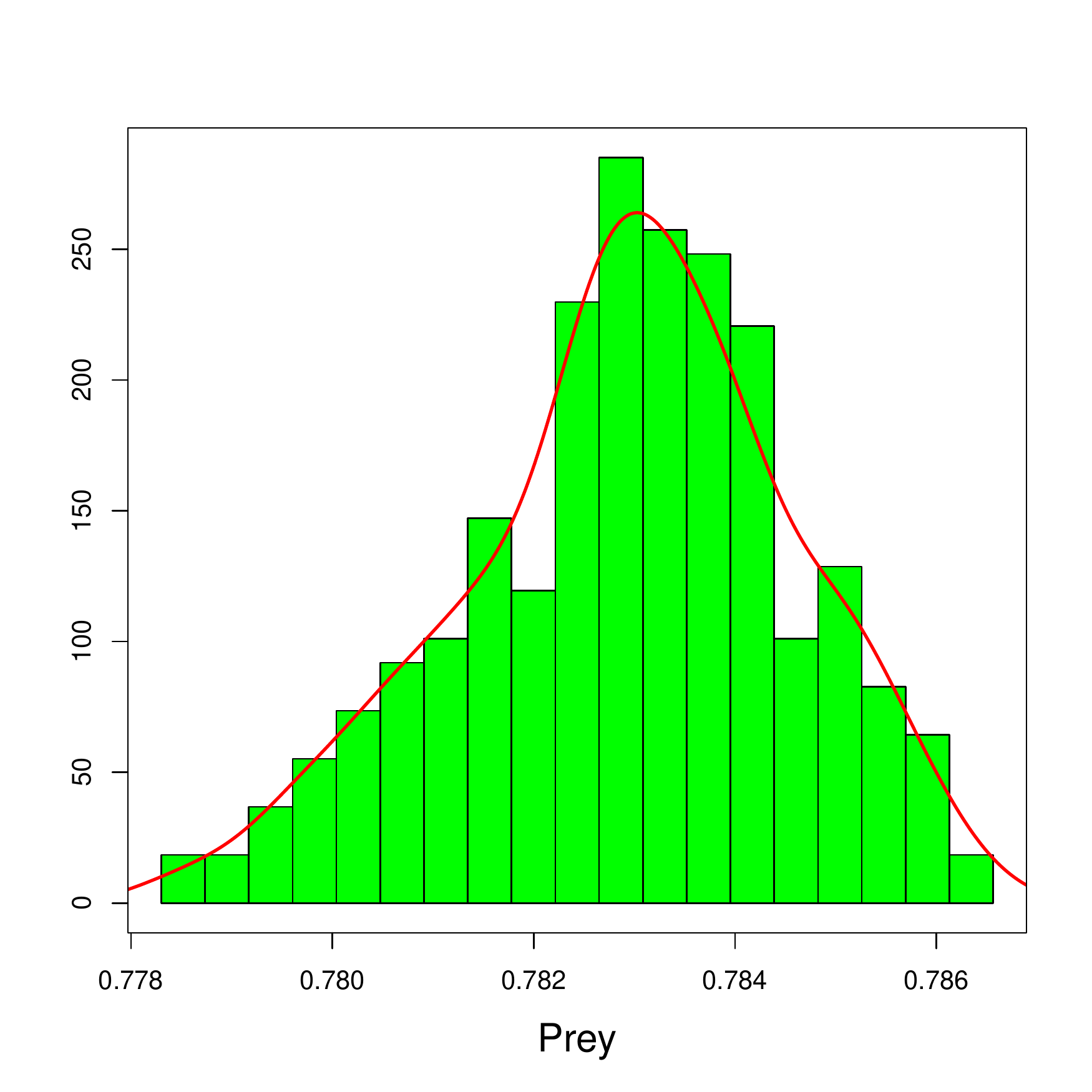} &   \includegraphics[width=2in, height=2.05in]{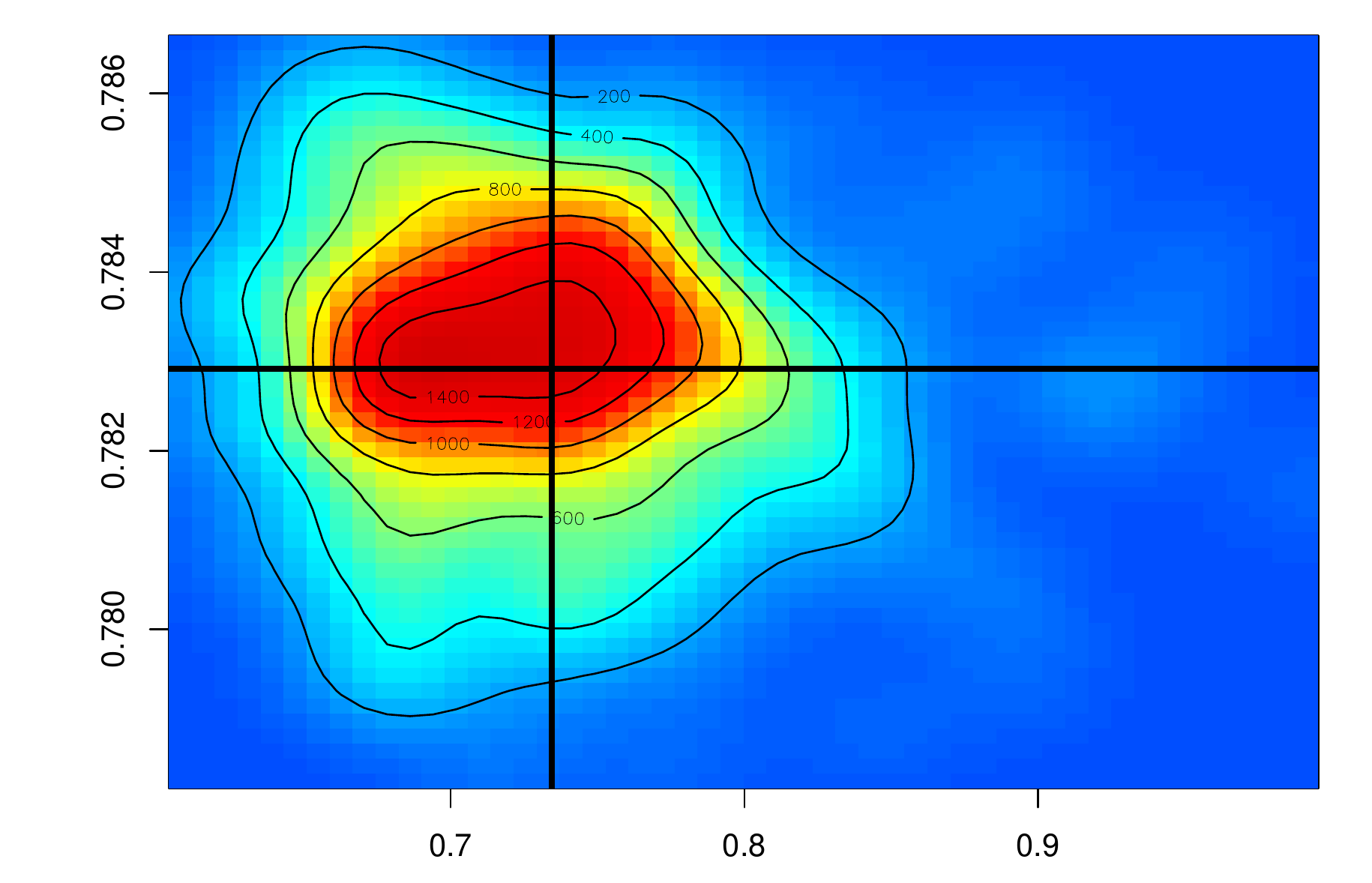} \\  
  &\includegraphics[scale=0.4]{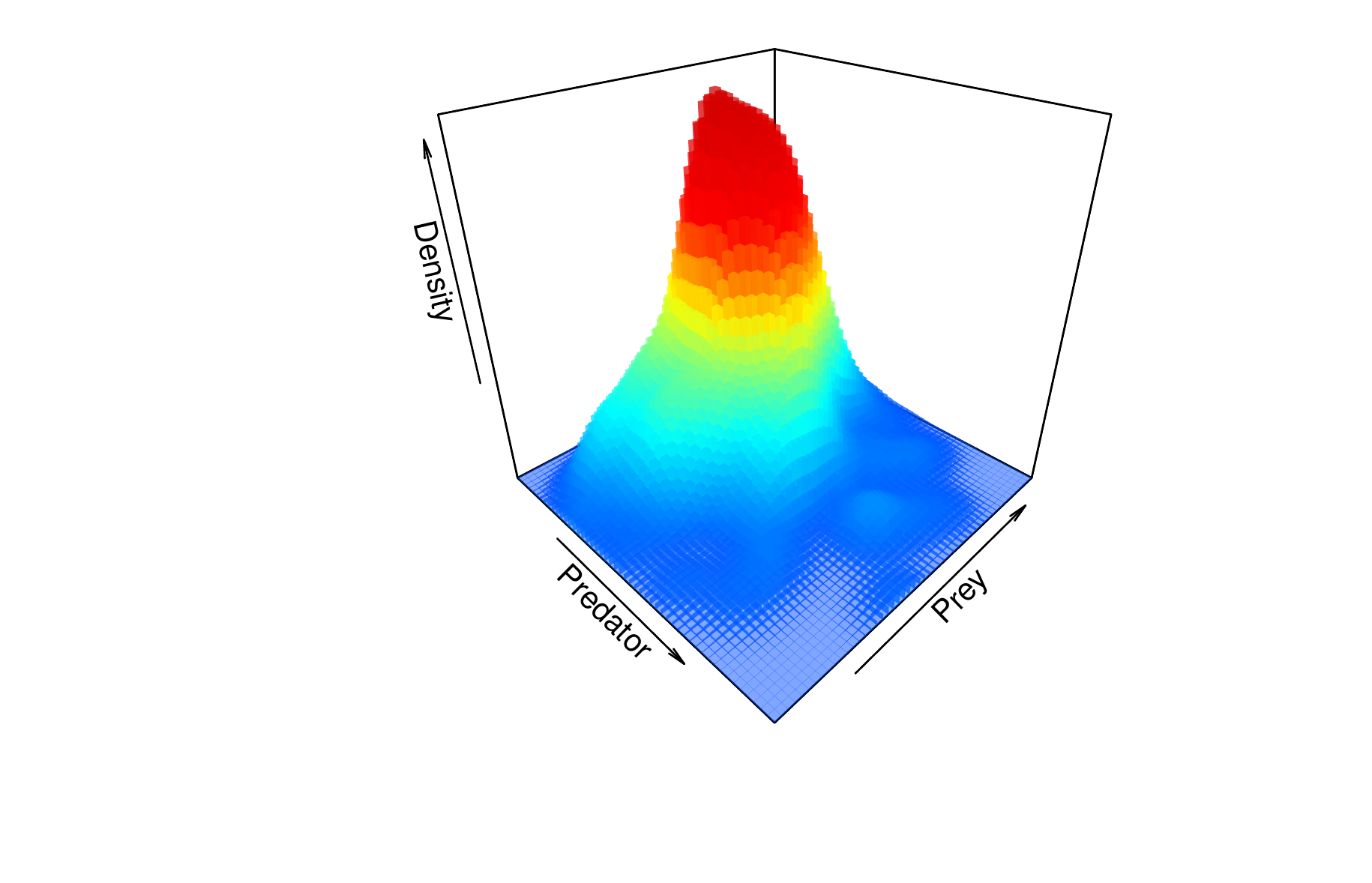} & \includegraphics[scale=0.32]{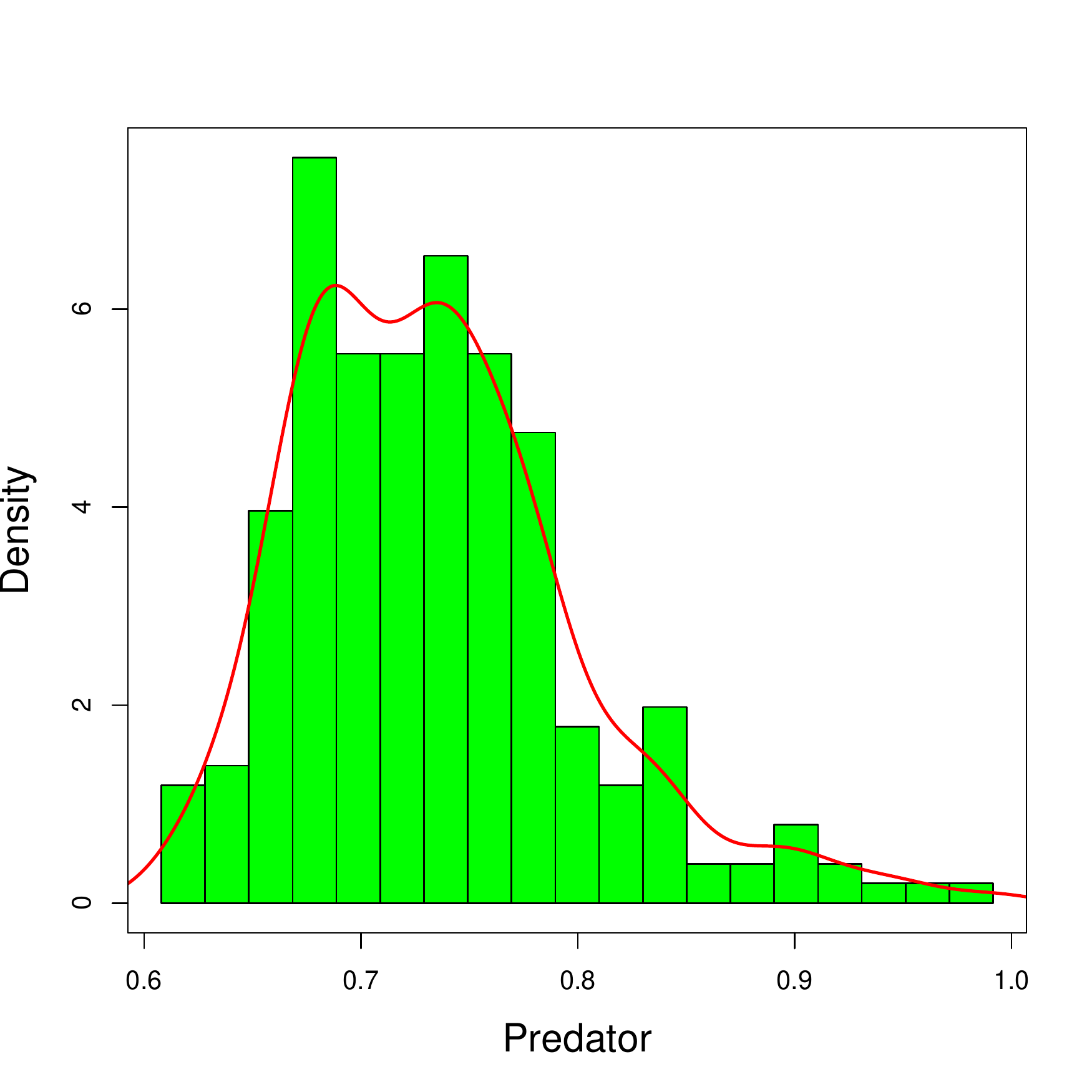} 
 \end{tabular}}
  \caption{In this case, we let $\sigma_1=6$ and $\sigma_2=0.1$.    In this case, the predator is persistent for trajectories above the deterministic unstable fixed point.}
 \label{Fig}
 \end{figure}

%%%%%%%%%%%%%%%%%%%%%%%%%%%%%%%%%%%%%%%%%%%%%%%%%%%%%%%%%%%%%
\subsubsection{\bf High stochasticity  on both species}
%%%%%%%%%%%%%%%%%%%%%%%%%%%%%%%%%%%%%%%%%%%%%%%%%%%%%%%%%%%%%

 \begin{figure}[H]
\centering \resizebox*{5.9in}{!}{
 \begin{tabular}{ccc}
 \includegraphics[scale=0.29]{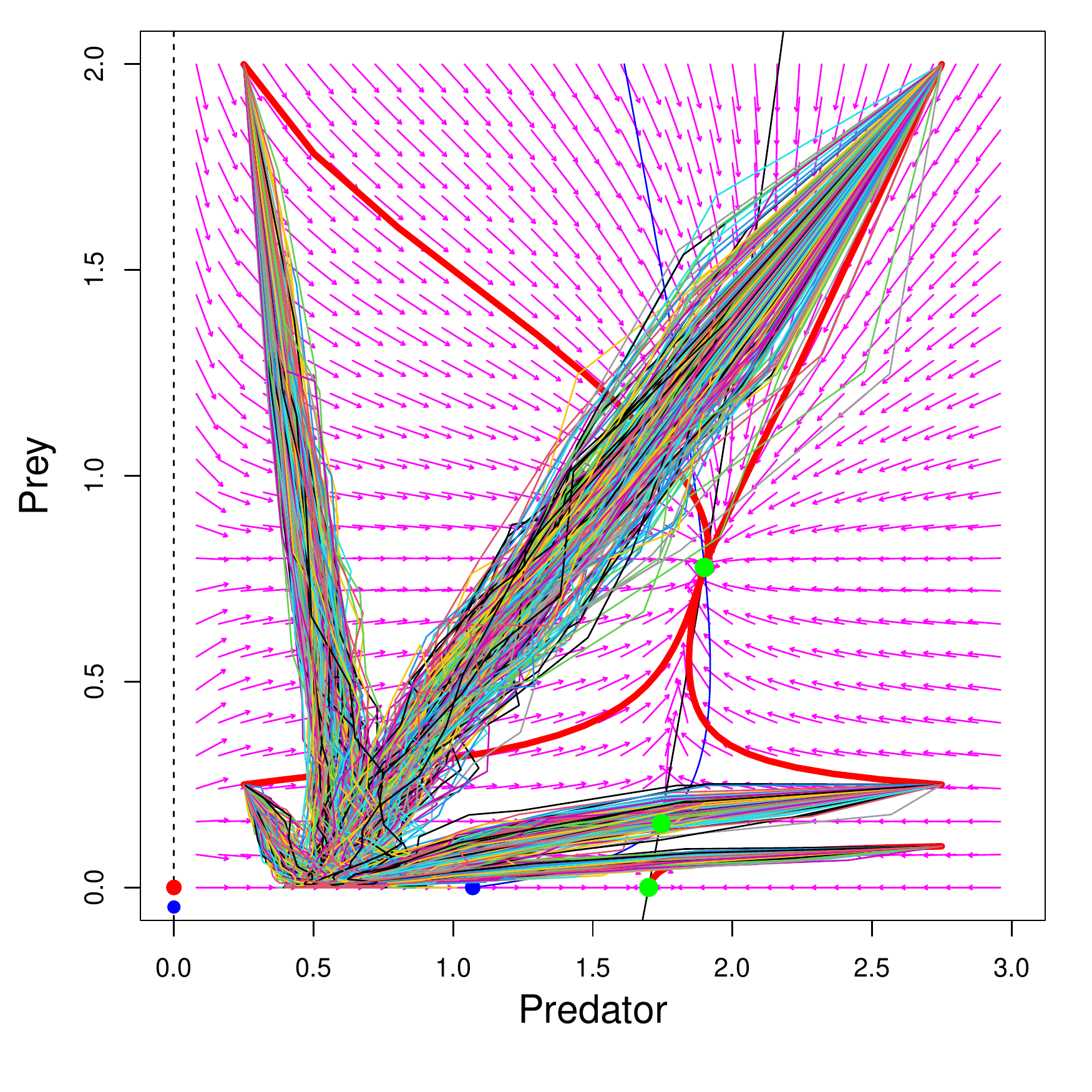} &  \includegraphics[scale=0.3, angle =90]{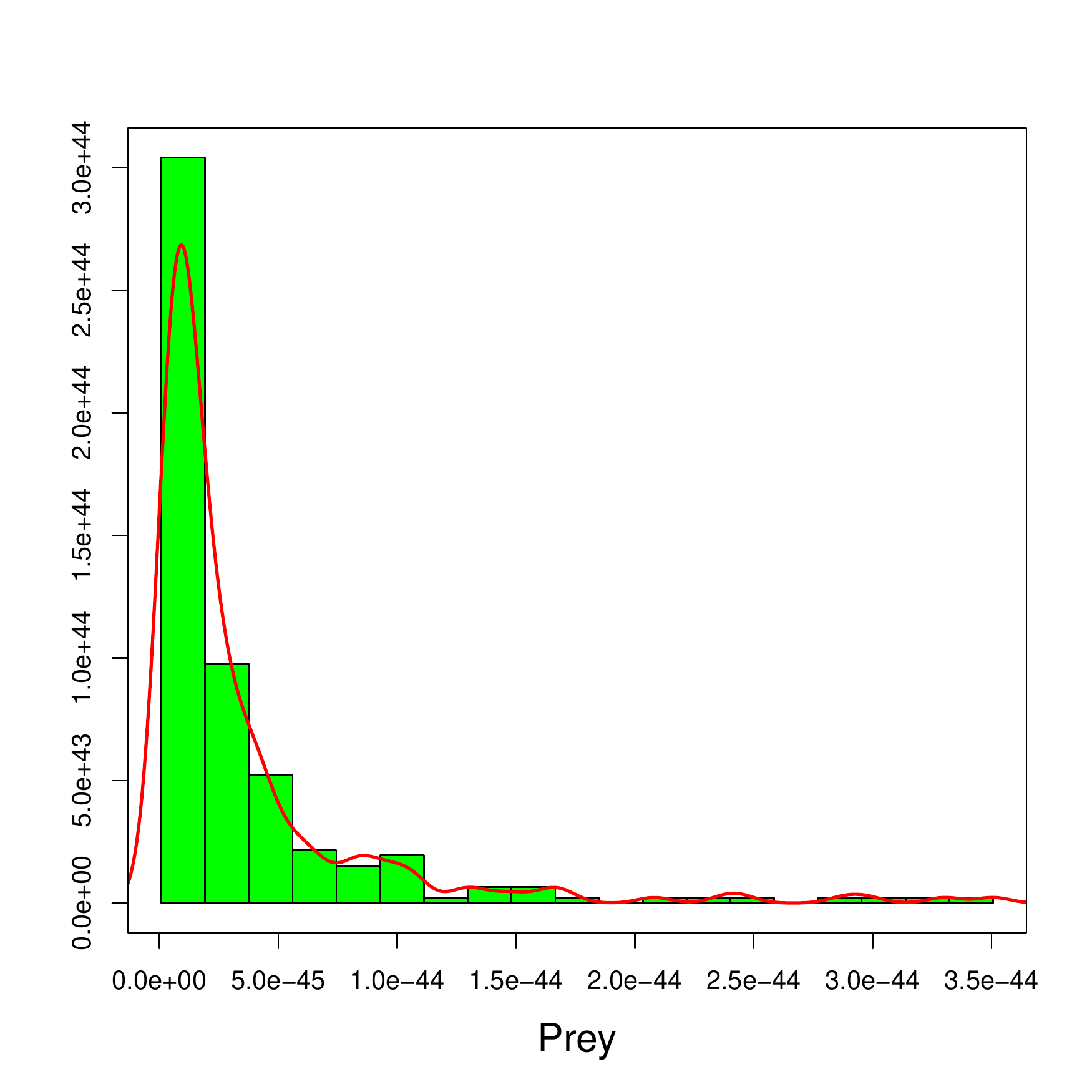} &   \includegraphics[width=2in, height=2.05in]{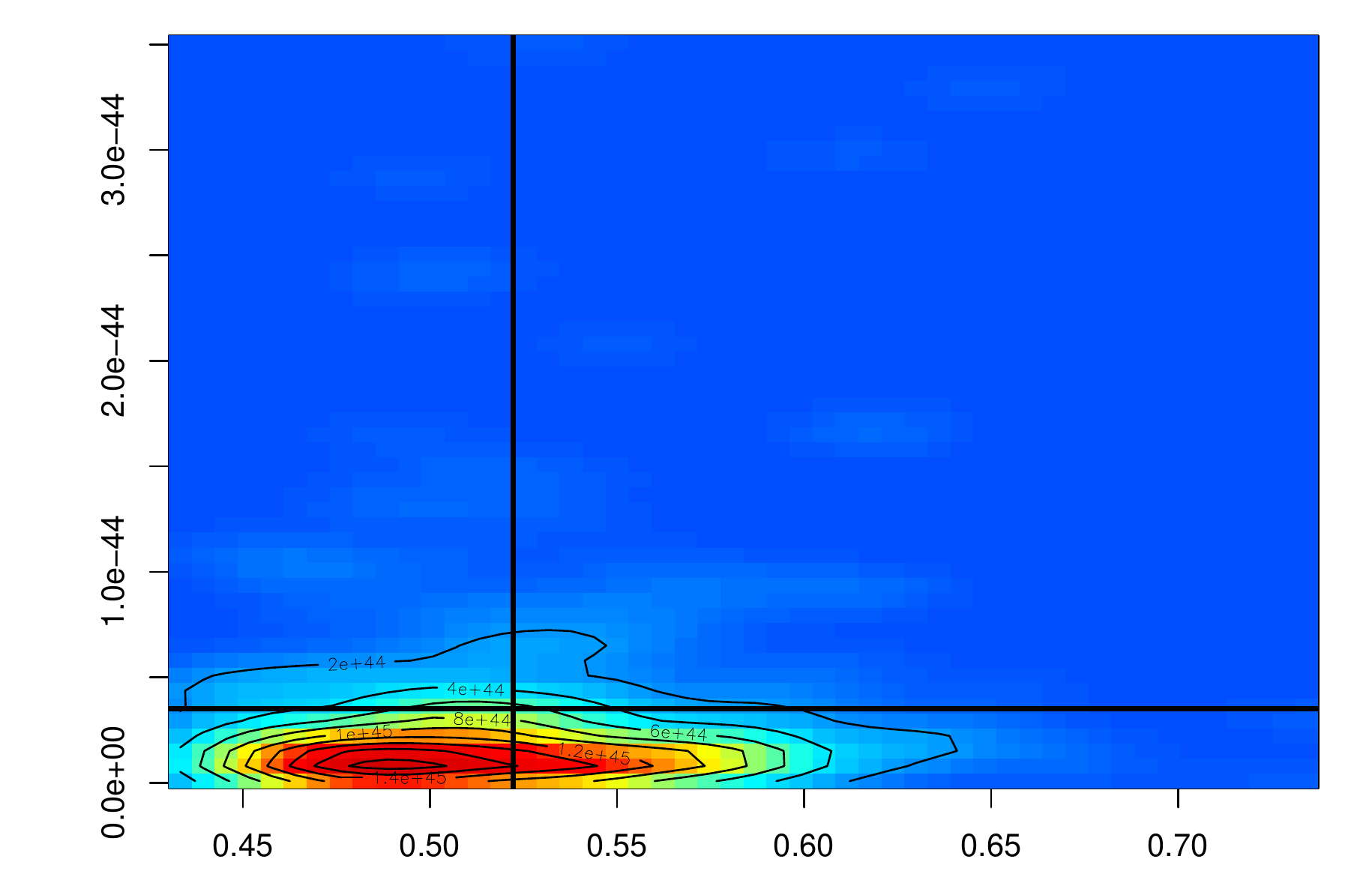} \\  
  &\includegraphics[scale=0.4]{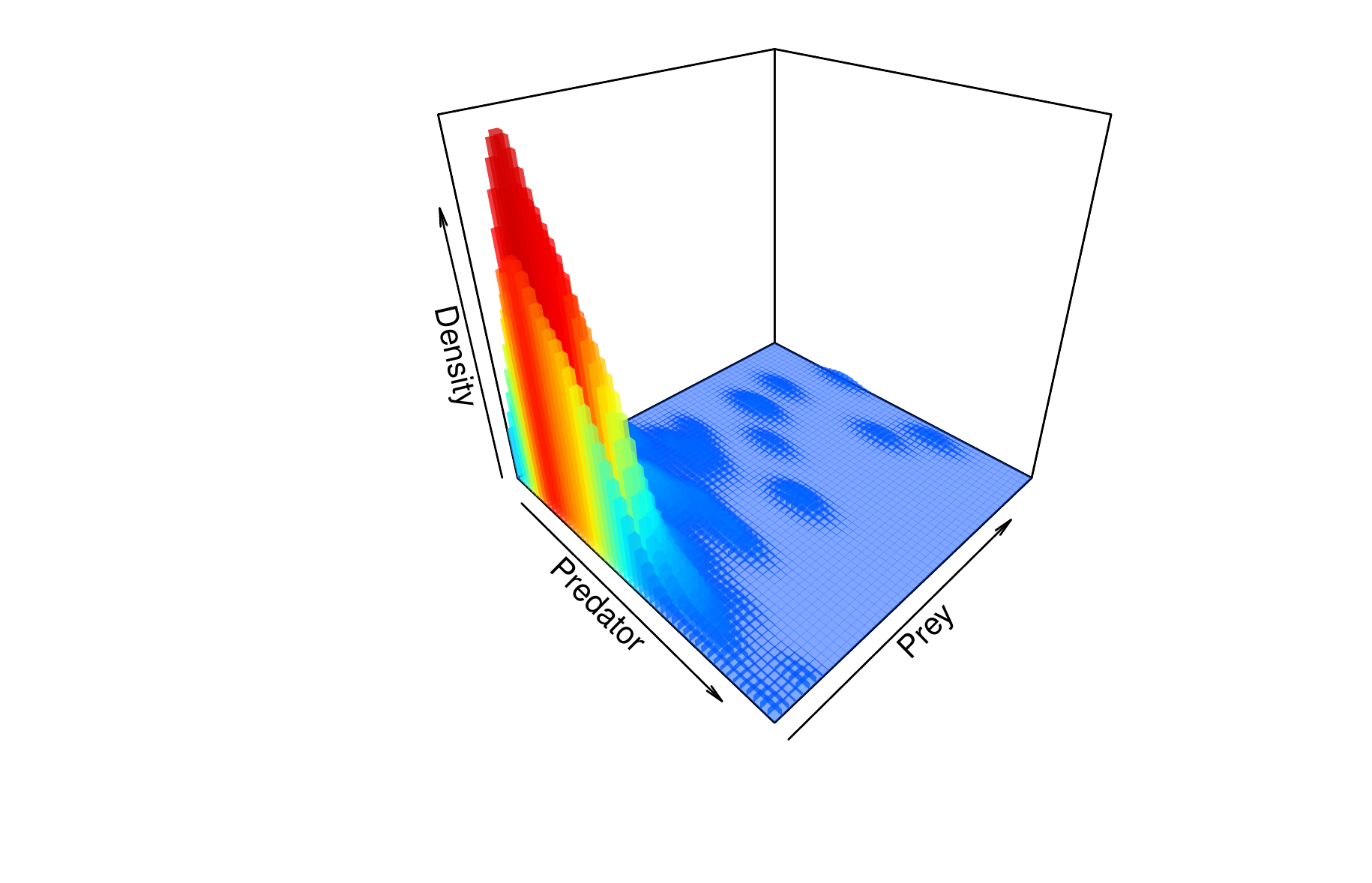} & \includegraphics[scale=0.32]{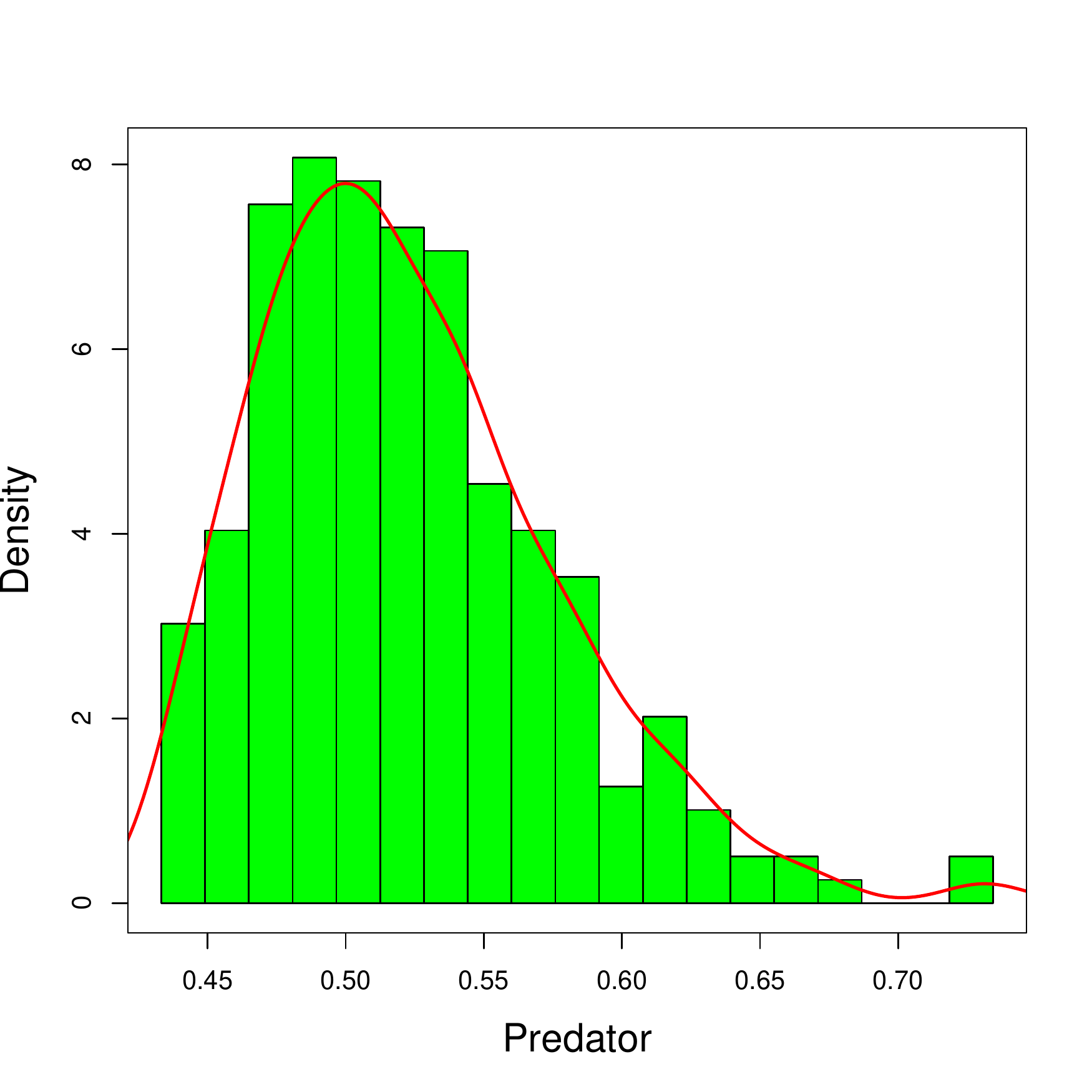} 
 \end{tabular}}

  \caption{In this case, we let $\sigma_1=6$ and $\sigma_2=6$. In this case, both species  will eventually goes extinct, though the species will be extinct  quicker than the predator.}
 \label{Fig}
 \end{figure}
To understand the effect of stochasticity on the densities of the species for different system parameters, one could track the change in density of predator and prey from successive increments of parameters $\sigma_1$ and $\sigma_2$. We can then assess the Wasserstein distance (see \cite{Villani2008}) between the these consecutive distribution to assess how similar or dissimilar they are. This would provide a threshold for stochasticity. At the same time, one could also track the average densities of predator and prey. Let $p>1$ be a real number. Let us recall that the Wasserstein distance $W(\theta_1,\theta_2)$ on a metric space $(\mathcal{M},D)$ between two probability measures $\theta_1$ and $\theta_2$ with  joint probability measure $\theta$ defined  on a set $\Gamma(\theta_1,\theta_2)$ is given as 
\[W_p(\theta_1,\theta_2)=\rb{\inf_{\theta \in \Gamma(\theta_1,\theta_2)}\set{ \int_{\mathcal{M}\times \mathcal{M}}\theta(u,v)[D(u,v)]^pdudv}}^{\frac{1}{p}}\;.\]
For practical purposes, distributions $\theta_1$ and $\theta_2$ will be taken as the  empirical probability measures with samples $U=(U_{(1)},U_{(2)},\cdots, U_{(K)})$ and $V=(V_{(1)},V_{(2)},\cdots, V_{(K)})$ from order statistics. 
Therefore,  an estimator of $W_p(\theta_1,\theta_2)$ is \[\widetilde{W}_p(\theta_1,\theta_2)=\rb{\sum_{i=1}^K\abs{U_{(i)}-V_{(i)}}^p}^{\frac{1}{p}}\;.\]
In our case, we will use $p=2$.  For  the predator, we use the samples $U_1=X(\sigma_{1,i})$ and $V_1=X(\sigma_{1,i+1})$ corresponding to  a  given stochastic parameter $\sigma_{1,i}\in [0.1,8]$ for $i=1,2,\cdots, K$.  Likewise, for the prey we use the samples $U_2=Y(\sigma_{2,i})$ and $V_2=Y(\sigma_{2,i+1})$, for $\sigma_{2,i}\in [0.1,8]$. 
In the figures below, we chose the starting point of the trajectories to be $x_0=5,y_0=1$ and $K=80$. For $i=1,\cdots, 80$. The parameters  $\sigma_{1,\cdot}=\sigma_{2,\cdot}=\sigma$ will be referred to   as stochasticity.  \noindent The left panels represent the Wasserstein distance between consecutive predator and prey distributions by stochasticity and the right panels stochasticity versus average population densities. 

%\noindent In  Figure \ref{fig-Wasserstein1} below,  $r_1=0.5,  r_2=0.1, H_1=1.2,H_2=1.8, m=0.4, a=0.1, b=0.9, c=0.01;d=0.3; p=0.1$. corresponding to  two interior fixed points in the deterministic case.\\

% \noindent In Figure \ref{fig-Wasserstein2} below, $r_1=0.1,r_2=0.1,  H_1=0.9, H_2=2, m=0.1, a=0.05, b=0.9, c=0.01, d=2.5, p=1.5$
%corresponding to a single interior  fixed point in the deterministic case.\\

% \noindent In Figure \ref{fig-Wasserstein3} below, $r_1=0.9;, r_2=0.1, H_1=-1, H_2=3.5, m=0.1, a=0.5, b=0.2, c=0.4, d=2.5,  p=0.1$ corresponding to the predator-free axial fixed point in the deterministic case.\\
% 
% \noindent  In Figure \ref{fig-Wasserstein4} below, $r_1=0.1, r_2=0.8, H_1=0.5, H_2=-.9, m=0.8, a=1.2, b=0.4, c=0.6, d=1, p=10$ corresponding to the prey-free axial fixed point in the deterministic case.\\
% 
%\noindent  In Figure \ref{fig-Wasserstein5} below,  $r_1=0.1,r_2=0.8, H_1=-0.1, H_2=-1, m=0.9, a=0.4, b=0.5, c=1, d=2.2, p=0.2$. corresponding to no interior fixed point in the deterministic case.\\
\begin{figure}[H]
\centering \includegraphics[scale=0.65]{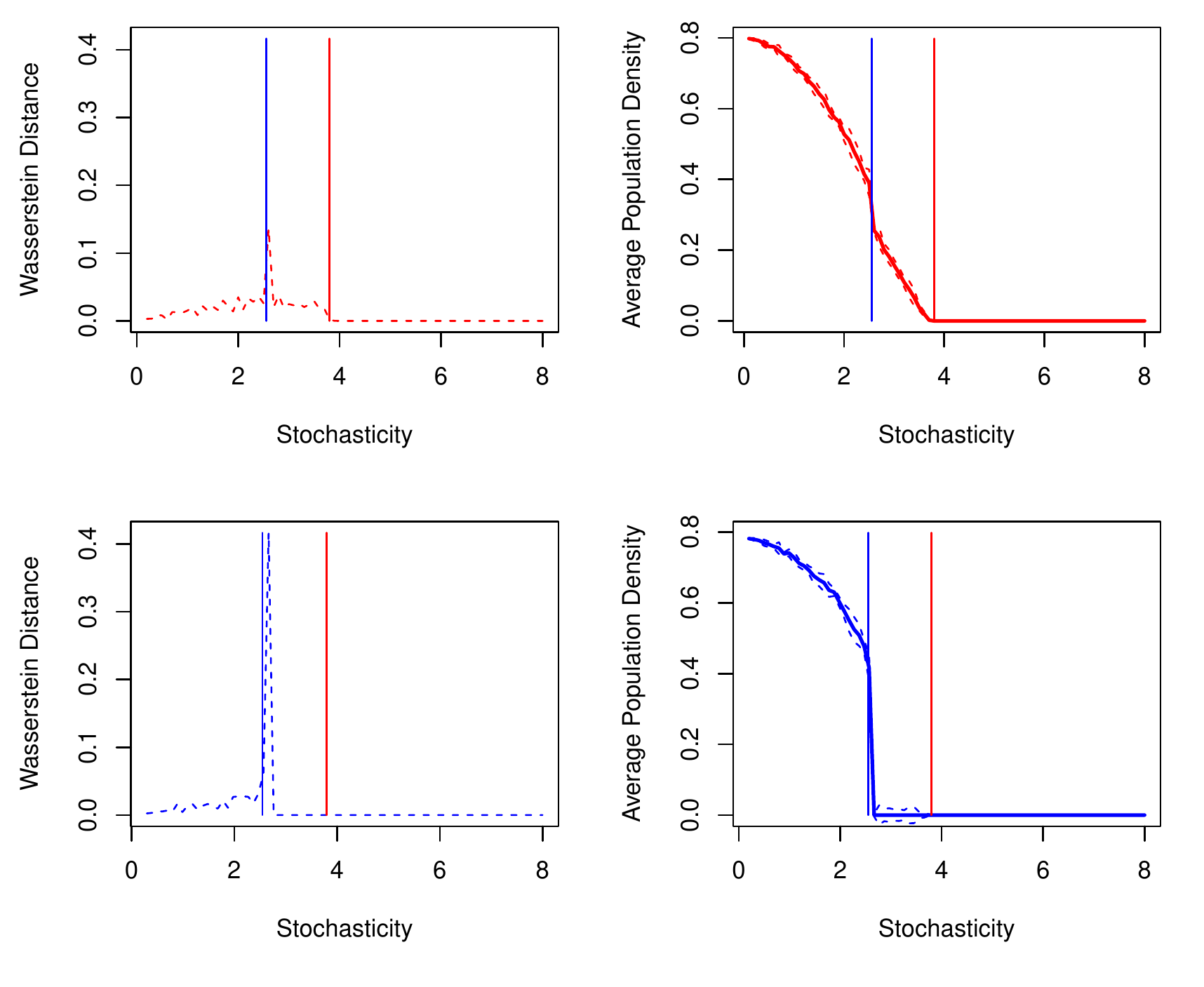}
\caption{ Here, $r_1=0.5,  r_2=0.1, H_1=1.2,H_2=1.8, m=0.4, a=0.1, b=0.9, c=0.01;d=0.3; p=0.1$. corresponding to  two interior fixed points in the deterministic case.
%This shows that for values of stochasticity close to each other, the distributions are very similar except before and after a stochasticity value of about 2.55, represented by the solid black line.  In the  plot on the right, dashed lines represent 95\% prediction intervals. It shows that increased stochasticity decreases population densities on average. This decrease accelerates precipitously  after the value of stochasticity  of about 2.55 found in the previous plot. This seems to indicates that stochasticity parameters $\sigma_1$ is a pitchfork  bifurcation parameters for the stochastic model.
}
\label{fig-Wasserstein1}
\end{figure}
\begin{figure}[H]
\centering \includegraphics[scale=0.65]{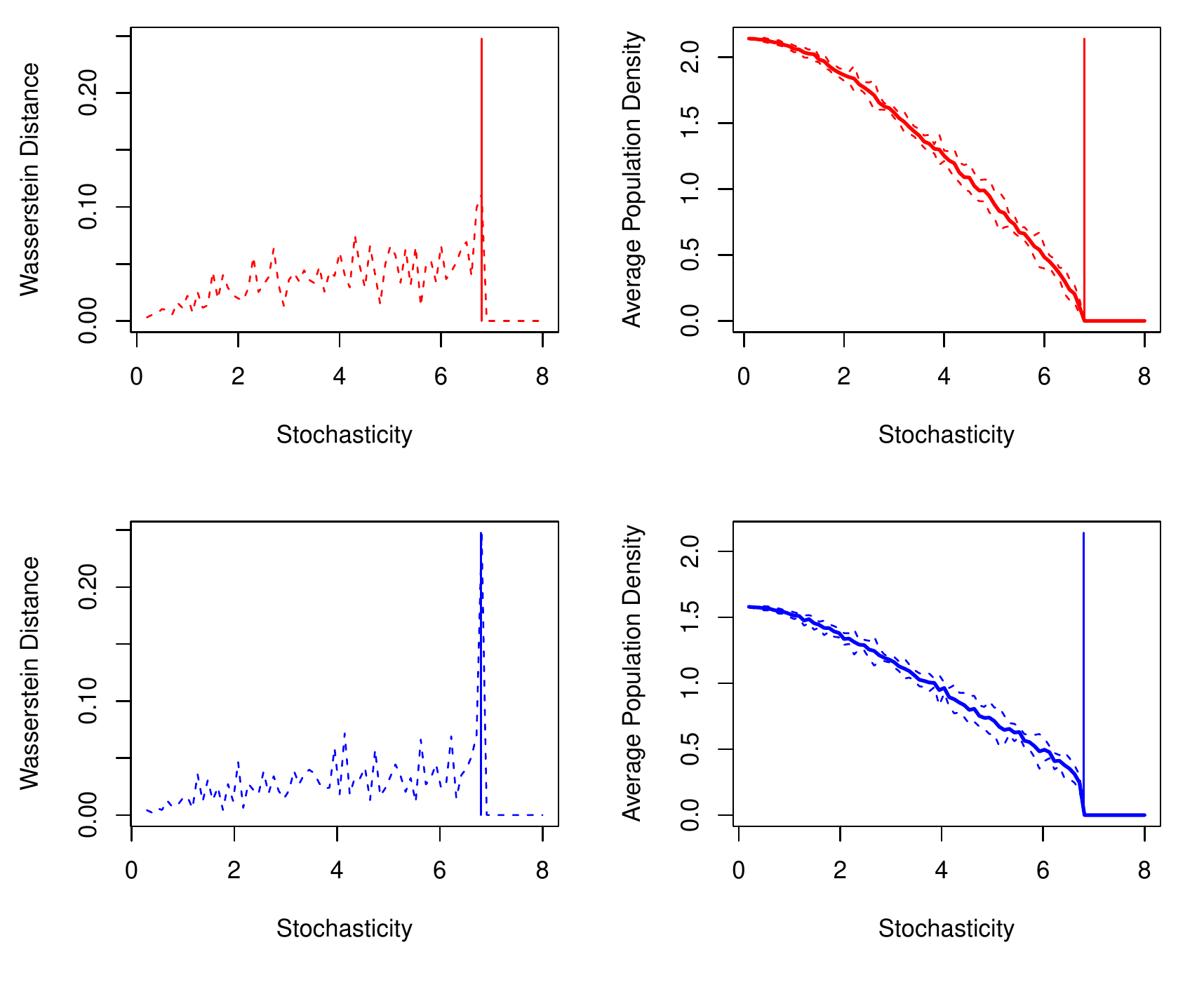}
\caption{$r_1=0.1,r_2=0.1,  H_1=0.9, H_2=2, m=0.1, a=0.05, b=0.9, c=0.01, d=2.5, p=1.5$
corresponding to a single interior  fixed point in the deterministic case.}
%\caption{The figure on the right represents the Wasserstein distance between consecutive predator and prey distributions by stochasticity for the parameters $r1=.5; H1=1.2; r2=.1;H2=1.8;m=.4;a=.1;b=.9;c=.01;d=.3; p=.1$. This shows that for values of stochasticity close to each other, the distributions are very similar except before and after a stochasticity value of about 2.55, represented by the solid black line.  In the  plot on the right, dashed lines represent 95\% prediction intervals. It shows that increased stochasticity decreases population densities on average. This decrease accelerates precipitously  after the value of stochasticity  of about 6.15 found in the previous plot. This seems to indicates that stochasticity parameters $\sigma_1$ is a pitchfork  bifurcation parameters for the stochastic model.}
\label{fig-Wasserstein2}
\end{figure}

\begin{figure}[H]
\centering \includegraphics[scale=0.65]{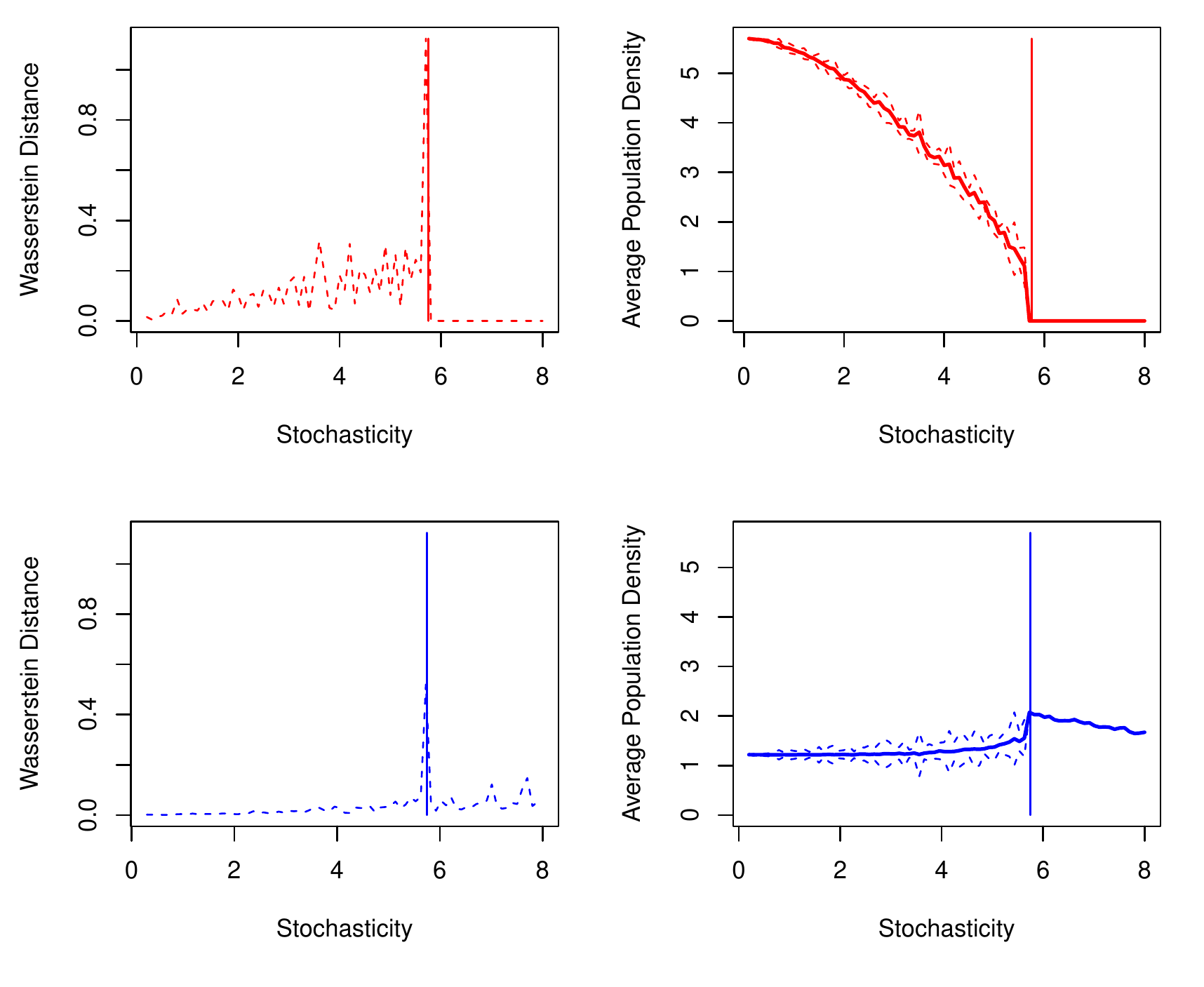}
\caption{$r_1=0.9;, r_2=0.1, H_1=-1, H_2=3.5, m=0.1, a=0.5, b=0.2, c=0.4, d=2.5,  p=0.1$ corresponding to the predator-free axial fixed point in the deterministic case.
%The figure on the right represents the Wasserstein distance between consecutive predator and prey distributions by stochasticity for the parameters $r1=.5; H1=1.2; r2=.1;H2=1.8;m=.4;a=.1;b=.9;c=.01;d=.3; p=.1$. This shows that for values of stochasticity close to each other, the distributions are very similar except before and after a stochasticity value of about 2.55, represented by the solid black line.  In the  plot on the right, dashed lines represent 95\% prediction intervals. It shows that increased stochasticity decreases population densities on average. This decrease accelerates precipitously  after the value of stochasticity  of about 6.15 found in the previous plot. This seems to indicates that stochasticity parameters $\sigma_1$ is a pitchfork  bifurcation parameters for the stochastic model.
}
\label{fig-Wasserstein3}
\end{figure}

\begin{figure}[H]
\centering \includegraphics[scale=0.65]{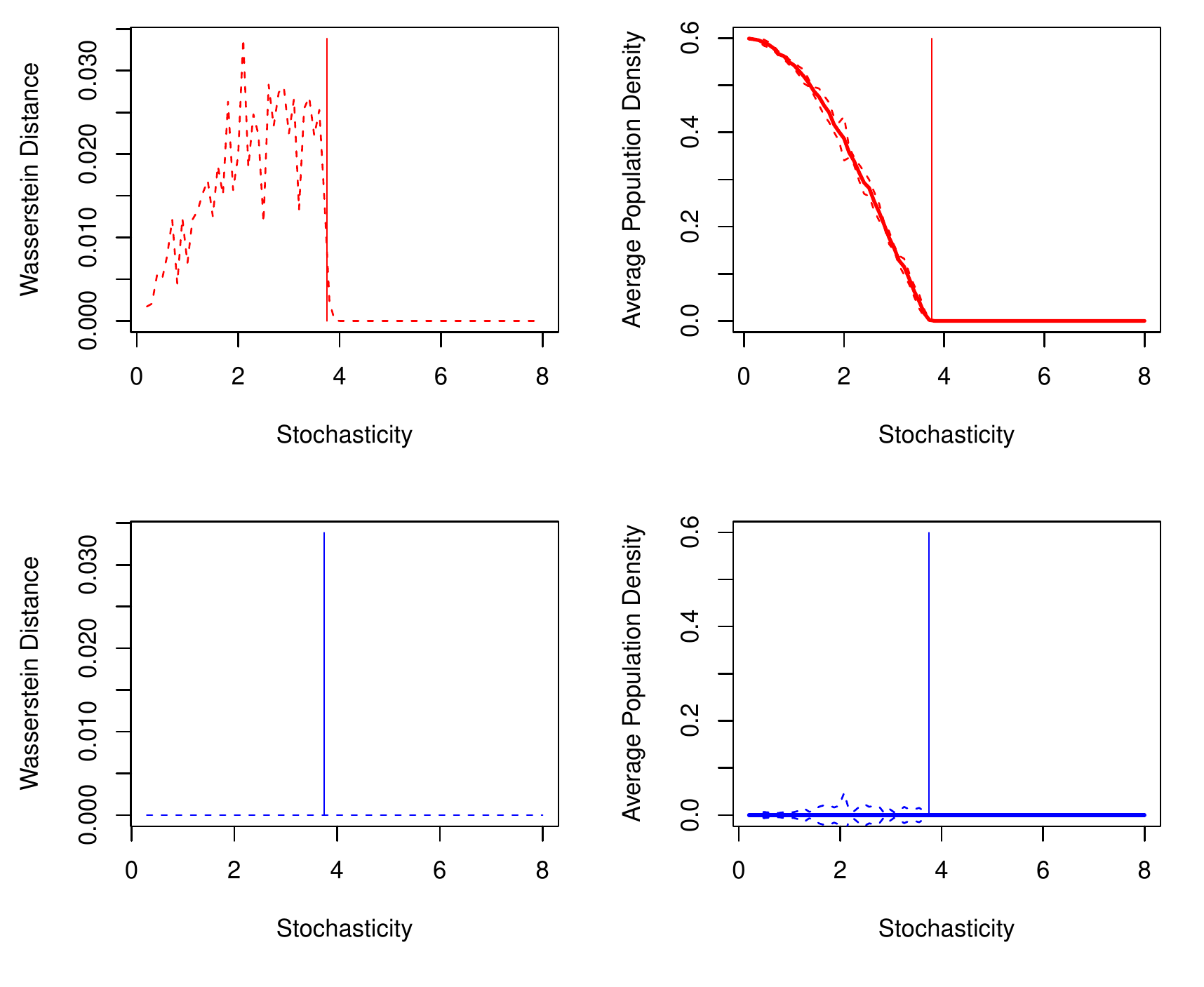}
\caption{ $r_1=0.1, r_2=0.8, H_1=0.5, H_2=-.9, m=0.8, a=1.2, b=0.4, c=0.6, d=1, p=10$ corresponding to the prey-free axial fixed point in the deterministic case.
%The figure on the right represents the Wasserstein distance between consecutive predator and prey distributions by stochasticity for the parameters $r1=.5; H1=1.2; r2=.1;H2=1.8;m=.4;a=.1;b=.9;c=.01;d=.3; p=.1$. This shows that for values of stochasticity close to each other, the distributions are very similar except before and after a stochasticity value of about 2.55, represented by the solid black line.  In the  plot on the right, dashed lines represent 95\% predi ction intervals. It shows that increased stochasticity decreases population densities on average. This decrease accelerates precipitously  after the value of stochasticity  of about 6.15 found in the previous plot. This seems to indicates that stochasticity parameters $\sigma_1$ is a pitchfork  bifurcation parameters for the stochastic model.
}
\label{fig-Wasserstein4}
\end{figure}

\begin{figure}[H]
\centering \includegraphics[scale=0.75]{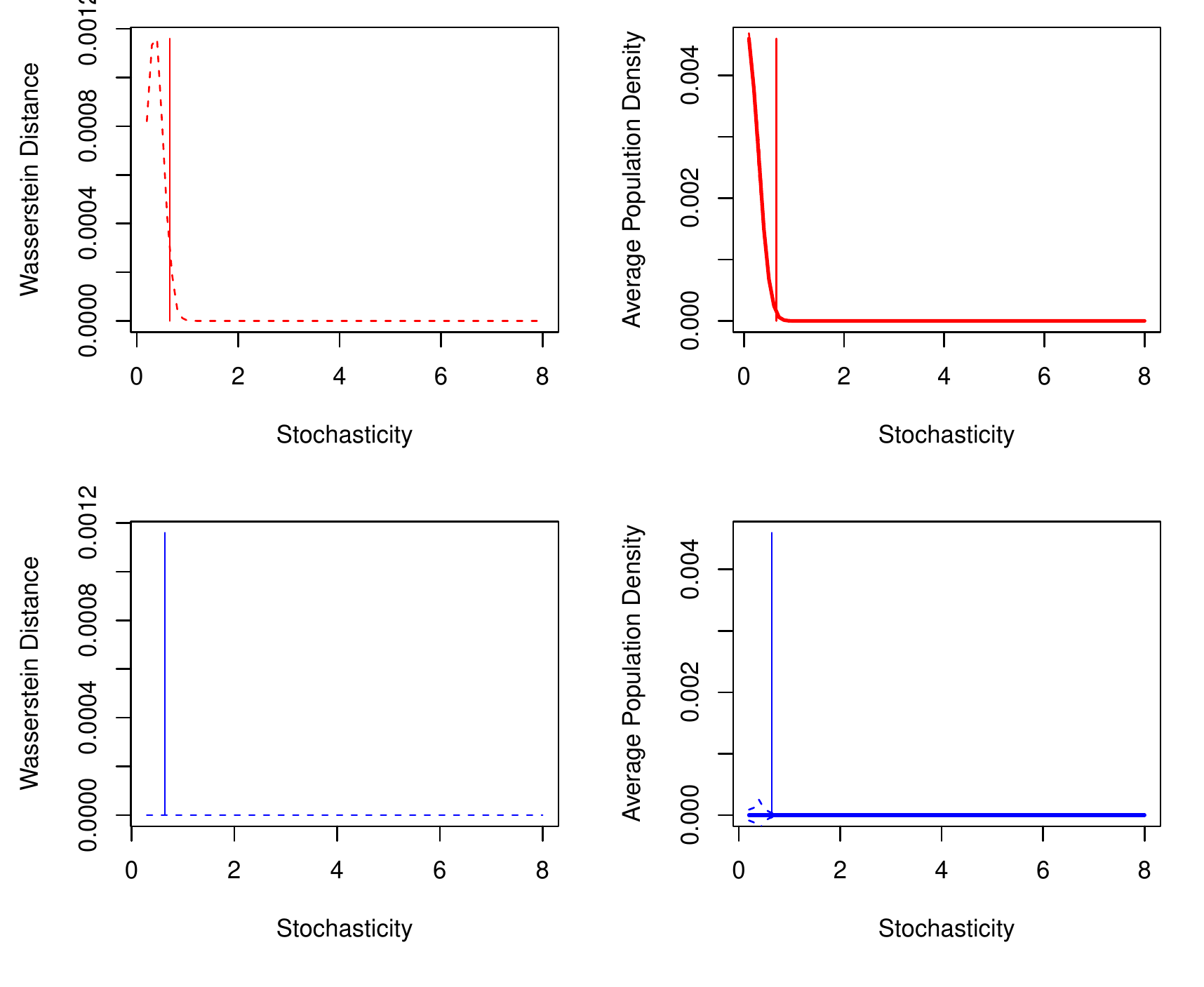}
\caption{ $r_1=0.1,r_2=0.8, H_1=-0.1, H_2=-1, m=0.9, a=0.4, b=0.5, c=1, d=2.2, p=0.2$. corresponding to no interior fixed point in the deterministic case.
%The figure on the right represents the Wasserstein distance between consecutive predator and prey distributions by stochasticity for the parameters $r1=.5; H1=1.2; r2=.1;H2=1.8;m=.4;a=.1;b=.9;c=.01;d=.3; p=.1$. This shows that for values of stochasticity close to each other, the distributions are very similar except before and after a stochasticity value of about 2.55, represented by the solid black line.  In the  plot on the right, dashed lines represent 95\% prediction intervals. It shows that increased stochasticity decreases population densities on average. This decrease accelerates precipitously  after the value of stochasticity  of about 6.15 found in the previous plot. This seems to indicates that stochasticity parameters $\sigma_1$ is a pitchfork  bifurcation parameters for the stochastic model.
}
\label{fig-Wasserstein5}
\end{figure}

%%%%%%%%%%%%%%%%%%%%%%%%%%%%%%%%%%%%%%%%%%%%%%%%%%%%%%%%%%%%%
%%%%%%%%%%%%%%%%%%%%%%%%%%%%%%%%%%%%%%%%%%%%%%%%%%%%%%%%%%%%%
\section{Discussion} \label{sect:discussion}
%%%%%%%%%%%%%%%%%%%%%%%%%%%%%%%%%%%%%%%%%%%%%%%%%%%%%%%%%%%%%
%%%%%%%%%%%%%%%%%%%%%%%%%%%%%%%%%%%%%%%%%%%%%%%%%%%%%%%%%%%%%

\begin{enumerate}
\item It is interesting to note the similarities in both Wasserstein distance plots and average population densities plots. 
\item Indeed, in  Figures \ref{fig-Wasserstein1}--\ref{fig-Wasserstein5}, there are critical values for stochasticity from which quantitatively, the average densities of both predator and prey change. They are represented by the solid vertical lines, red for the predator and blue for the prey.  This is matched in  Wasserstein plots with a drop in distance. 
\item From a purely bifurcation point of view, this drop suggests that stochasticity seems to produce a supercritical pitchfork bifurcation in both species. 
\item From a purely statistical point of view, the Wasserstein distance obtained is clearly unimodal and right-skewed  when a bifurcation occurs with the right skew being clearly zero.
\item Another important takeaway is that persistence of the prey is clearly seen through the Wasserstein distance, see Figure \ref{fig-Wasserstein3}. In fact, the Wasserstein is unimodal with values in the same range except for a single peak due to the disappearance of the predator.
\item Figure \ref{fig-Wasserstein4} is interesting in that it shows that while the predator persists for while under environmental stochasticity, it eventually goes extinct due to lack of prey. On the other, the prey quickly goes extinct due to both predation and environmental fluctuations.
\item Is important to note that the simulations were done with the parameters of the deterministic case corresponding to the two interior fixed points. Similar observations as above can be drawn from the remaining cases.
\item We observe that the per-capita magnitude of the environmental fluctuations are just linear functions of the populations densities. More structured per-capita magnitude functions  can be considered, especially polynomial of higher degree without fundamentally altering the conclusions we obtained here.
%\begin{itemize}
%\item When higher one the prey than on the predator, it is beneficial to the prey.
%\item When higher one the predator than on the prey, it is  deleterious to the predator.
%%\item When moderate on both species, it affect the species similarly and may prevent extinction of both species.
%%\item When high on both species, it seems to benefit the prey more.
%%\item However, if the perturbations at each period are of the form $(x_k^2N_{1,k}(\Delta t), y_k^2N_{2,k}(\Delta t))$, then unstable interior equilibrium becomes "periodically" stable.
%\end{itemize}
%\item Simulations of the stochastic model show that in fact, as one would expect in the real world, environmental changes affect the species at various degrees depending on the genetic predisposition to withstand them.
%%\centering\includegraphics[scale=0.4]{XsquaredNoise.pdf}
\end{enumerate}
Our simulations are just an initial proof of concept that  the Wasserstein distance may be used successfully to study bifurcation in stochastic environments. However, further investigations, both theoretically and practically  are needed establish a more solid understanding. Since the Wasserstein distance  is easy to calculate, it may be worthwhile to see if it can be used to empirically assess  chaotic events as an alternative to calculating the largest Lyapunov exponent.  
\bibliography{HierarchicalHollinsII}

%%%%%%%%%%%%%%%%%%%%%%%%%%%%%%%%%%%%%%%%%%%%%%%%%%%%%%%%%%%%
\section{Appendix}  \label{sect: Append}
%%%%%%%%%%%%%%%%%%%%%%%%%%%%%%%%%%%%%%%%%%%%%%%%%%%%%%%%%%%%

\subsection{Proof of Theorem \ref{thm:GlobalStat}}
\begin{proof}
Let \[F(x,y)=\rb{xe^{\ds r_1-x+\frac{d(b+cx)y}{1+py(b+cx)}+H_1}, \quad ye^{\ds r_2-y-\frac{m}{a+y}-\frac{(b+cx)y}{1+py(b+cx)}+H_2}}\;.\]
  Consider the Lyapunov function $L(x,y)=x^2+y^2$.\\
  Then 
  \begin{eqnarray*}
  L(F(x,y))-L(x,y)&=& x^2\Intv{e^{\ds 2\rb{r_1+H_1-x+\frac{d(b+cx)y}{1+py(b+cx)}}}-1}\\
  &+& y^2\Intv{e^{\ds 2\rb{r_2+H_2-\frac{m}{a+y}-\frac{(b+cx)y}{1+py(b+cx)}}}-1}\\
  &\leq&   x^2\Intv{e^{\ds 2\rb{r_1+H_1+\frac{ db}{p}}}-1}+y^2\Intv{e^{\ds 2\rb{r_2+H_2}}-1}\;.
  \end{eqnarray*}
  Thus, $L(F(x,y))-L(x,y)\leq 0$ if $\ds r_1+H_1+\frac{ d}{p}<0$ and $r_2+H_2<0$\;. This shows that the equilibrium point $(0,0)$ is globally asymptotic stable.  Recall that  
  \begin{eqnarray*}
h(x,y)&=&\frac{(b+cx)y}{1+py(b+cx)}\leq \frac{d}{p}\;,\\
F(x,y)&=&(F_1(x,y), F_2(x,y))=(xf_1(x,y), yf_2(x,y))\;.
\end{eqnarray*} and 
\begin{equation}\label{eqn:boundedness}
\begin{cases}
f_1(x,y)=e^{\ds r_1+H_1-x+dh(x,y)}\leq e^{r1+H_1+\frac{d}{p}}\\
f_2(x,y)= e^{\ds r_2+H_2-y-\frac{m}{a+y}-h(x,y)}\leq e^{r_2+H_2}
\end{cases}
\end{equation}
%Put \[\mbox{\blue{$\alpha=e^{r1+H_1+\frac{db}{p}}$} and \blue{$\beta=e^{r_2+H_2}$}}\;.\]
Consider the Lyapunov function \[L(x,y)=(x-x_*)^2+(y-y_*)^2\;.\]
and let $\Delta L=L(F(x,y))-L(x,y)\;.$
Then
\begin{eqnarray*}
\Delta L &=& (xf_1(x,y)-x_*)^2+(yf_2(x,y)-y_*)^2-(x-x_*)^2-(y-y_*)^2\\
&\leq & (\alpha x-x_*)^2+(\beta x-y_*)^2-(x-x_*)^2-(y-y_*)^2\\ 
&=& (\alpha^2-1)x^2-2(\alpha-1)x_*x+(\beta^2-1)y^2-2(\beta-1)y_*y\\
&=& (\alpha^2-1)\rb{x-\frac{1}{\alpha+1}x_*}^2-\frac{\alpha-1}{\alpha+1}x_*^2\\
&+& (\beta^2-1)\rb{y-\frac{1}{\beta+1}y_*}^2-\frac{\beta-1}{\beta+1}y_*^2
%&\leq & (\alpha -1)^2x^2-(\alpha-1)x+ (\beta -1)^2y^2-(\beta-1)y\\
\end{eqnarray*}
Let $a_0=\max\set{\alpha^2-1,\beta^2-1}$ and let $\ds b_0=\frac{\alpha-1}{\alpha+1}x_*^2+\frac{\beta-1}{\beta+1}y_*^2$.
If $a_0\leq 0$and $x_*=y_*=0$ then $\Delta L\leq 0$.
In particular, if $a_0=0$, then clearly, for any $(x_*,y_*)$, we have $\Delta L\leq 0$.\\
Suppose $a_0>0$. Then  $b_0\geq 0$. Put $r_0^2=\frac{b_0}{a_0}$. \\Then $\Delta L\leq 0$ if $\rb{x-\frac{1}{\alpha+1}x_*}^2+\rb{y-\frac{1}{\beta+1}y_*}^2<r_0^2$. That is, $\Delta L\leq 0$ if  $(x,y)\in B\rb{\rb{\frac{1}{\alpha+1}x_*,\frac{1}{\beta+1}y_*},r_0}$,  hence  local stability.
\end{proof}

%%%%%%%%%%%%%%%%%%%%%%%%%%%%%%%%%%%%%%%%%%%%%%%%%%%%%%%%%%%%
\subsection{Proof of Theorem \ref{thm:MeanPersistence}}
%%%%%%%%%%%%%%%%%%%%%%%%%%%%%%%%%%%%%%%%%%%%%%%%%%%%%%%%%%%%

\begin{proof}
Let $0<\theta<1$. Then $e^{-x}\geq -x+\theta$.\\
Let $\ds g_1(x(t)):=\ln(x(t))$. Then using It\^{o}'s formula, we have 
\begin{eqnarray*} 
dg_1(x(t))&=&\Intv{\frac{\partial g_1}{\partial t}+x(t)f_1[x(t),y(t)]\cdot \frac{\partial g_1}{\partial x}+\frac{\partial^2 g_1}{\partial x^2}\cdot \frac{\sigma_1^2}{2}}dt+x(t)\sigma_1 \frac{\partial g_1}{\partial x}dW_1(t)\\
&=&\Intv{e^{r_1+H_1-x(t)+\frac{d(b+cx(t))}{1+py(t)(b+cx(t))}}-\frac{\sigma_1^2}{2}}dt+\sigma_1dW_1(t)\\
&\geq& \Intv{e^{r_1+H_1-x(t)}-\frac{\sigma_1^2}{2}}dt+\sigma_1dW_1(t)\\
& \geq & \Intv{e^{r_1+H_1}\rb{-x(t)+\theta}-\frac{\sigma_1^2}{2}}dt+\sigma_1dW_1(t)\;.\\
\end{eqnarray*}
Integrating from 0 to $t$, we have 
\begin{eqnarray*}
\ln\rb{\frac{x(t)}{x(0)}}&\geq& \Intv{\theta e^{r_1+H_1}-\frac{\sigma_1^2}{2}}t-e^{r_1+H_1} \scp{x(t)}+\sigma_1W_1(t)\;.
\end{eqnarray*}
Let $\gamma_{11}=\theta e^{r_1+H_1}-\frac{1}{2}\sigma_1^2$ and $\gamma_{12}=e^{r_1+H_1}$. Let $\ds g_2(y(t)):=\ln(y(t))$.\\
Then using It\^{o}'s formula, we have 
\begin{eqnarray*}
dg_2(y(t))&=&\Intv{\frac{\partial g_1}{\partial t}+y(t)f_2[x(t),y(t)]\cdot \frac{\partial g_1}{\partial y}+\frac{\partial^2 g_1}{\partial y^2}\cdot \frac{\sigma_2^2}{2}}dt+y(t)\sigma_1 \frac{\partial g_1}{\partial y}dW_1(t)\\
&=&\Intv{e^{r_2+H_2-y(t)-\frac{m}{a+y(t)}-\frac{b+cx(t)}{1+py(t)(b+cx(t))}}-\frac{1}{2}\sigma_2^2}dt+\sigma_2dW_2(t)\\
&\geq& \Intv{e^{r_2+H_2-y(t)-\frac{m}{a}-\frac{1}{p}}-\frac{1}{2}\sigma_2^2}dt+\sigma_2dW_2(t)\\
& \geq & \Intv{e^{r_2+H_2-\frac{m}{a}-\frac{1}{p}}(-y(t)+\theta)-\frac{1}{2}\sigma_2^2}dt+\sigma_2dW_2(t)\;.\\
\end{eqnarray*}
Integrating from 0 to $t$, we have 
\begin{eqnarray*}
\ln\rb{\frac{y(t)}{y(0)}}&\geq& \Intv{\theta e^{r_2+H_2-\frac{m}{a}-\frac{1}{p}}-\frac{1}{2}\sigma_2^2}t-e^{r_2+H_2-\frac{m}{a}-\frac{1}{p}} \scp{y(t)}+\sigma_2W_2(t)\;.
\end{eqnarray*}
Let $\gamma_{21}=\theta e^{r_2+H_2-\frac{m}{a}-\frac{1}{p}}-\frac{1}{2}\sigma_2^2$ and $\gamma_{22}=e^{r_2+H_2-\frac{m}{a}-\frac{1}{p}}$\;.\\
It follows that $\gamma_{11},\gamma_{21}>0$ if $1>\theta>\frac{1}{2}\max \set{\sigma_1^2e^{-r_1-H_1},\sigma_2^2e^{\frac{m}{a}+\frac{1}{p}-r_2-H_2}}$.\\
For such $\theta$,  put $\gamma_1=(\gamma_{11},\gamma_{12})$ and  $\gamma_2=(\gamma_{21},\gamma_{22})$. Let $\sigma=(\sigma_1,\sigma_2)$.\\
Since $X(t)=(x(t),y(t))$ and $W(t)=(W_1(t),W_2(t))$, we will have in vector form

\[\ln(X(t))\geq\gamma_1 t-\gamma_2\int_0^t X(u)du+\sigma W(t)\;.\]
Then by the Lemma \ref{lemMeanPers} above, we conclude that
\[\lim_{t \to \infty}\inf \scp{X(t)}\geq \frac{\gamma_1}{\gamma_2}>0\quad \mbox{almost surely}\;.\]
Therefore, $X(t)$ is strongly persistent in mean.

\end{proof}
%%%%%%%%%%%%%%%%%%%%%%%%%%%%%%%%%%%%%%%%%%%%%%%%%%%%%%%%%%%%
\subsection{Proof of Theorem \ref{thm:StatDist}}
%%%%%%%%%%%%%%%%%%%%%%%%%%%%%%%%%%%%%%%%%%%%%%%%%%%%%%%%%%%%

Consider the the stochastic differential equation \eqref{eqn:stochastic}. It is of the form 
\[dZ_i(t)=Z_i(t)f_i({\bf Z}(t))dt+Z_i(t)g_i({\bf Z}(t))dW_i(t), \quad i=1,2 \;,\]
where $Z_1(t)=x(t)$ and $Z_2(t)=y(t)$. We observe from Remark \ref{rem:Lyapunov:exp} that if \[\ds \frac{1}{2}\max\set{\sigma_1^2e^{-r_1-H_1},\sigma_2^2e^{\frac{m}{a}+\frac{1}{p}-r_2-H_2}}<1\;,\] then $\ds \max_{i=1,2} \lambda(\mu)>0$, for any ergodic invariant measure $\mu$. For given  a ${\bf z}\in \R^2$ and according to Theorem 2.1 in \cite{Hening2021}, it remains to  check the following assumptions
\begin{enumerate}
\item[$A_1:$] $\mbox{diag}(g_1({\bf z}),\cdots, g_n({\bf z}))\Gamma^T\Gamma \mbox{diag}(g_1({\bf z}),\cdots, g_n({\bf z}))$ is a positive definite matrix. 
\item[$A_2:$] $f_i(\cdot), g_i(\cdot): \R^{2,\circ}_+\to \R$ are locally Lipschitz functions, for $i=1,2$.
\item[$A_3:$] There exist ${\bf c}=(c_1,c_2)\in \R^{2,\circ}_+ $ and $\gamma_b>0$ such that 
\[\limsup_{\norm{z}\to \infty}\Intv{ \frac{\sum_{i=1}^2c_iz_if_i({\bf z})}{1+\sum_{i=1}^2 c_iz_i}-\frac{1}{2}\frac{\sum_{i,j=1}^2\sigma_{ij}c_ic_jz_iz_jg_i({\bf z})g_j({\bf z})}{\rb{1+\sum_{i=1}^n c_iz_i}^2}+\gamma_b\rb{1+\sum_{i=1}^2\rb{\abs{f_i({\bf z})+g_i^2({
\bf z})}}}}<0\;.\] 
\end{enumerate}
We note that in our case, $g_i({\bf Z}(t))=1$ for $i=1, 2$. \\
For $A_1$, we have
\[\mbox{diag}(g_1({\bf z}),\cdots, g_n({\bf z}))\Gamma^T\Gamma \mbox{diag}(g_1({\bf z}),\cdots, g_n({\bf z}))=\begin{pmatrix} \sigma_1 & 0\\ 0& \sigma_2\end{pmatrix}\;, \] which is a positive definite matrix since $\sigma_i>0$ for $i=1,2$.\\
For $A_2$, the $f_i(\cdot )$'s are locally continuously differentiable functions with bounded derivative,  therefore they are  locally Lipschitz functions.\\
 $A_3$ requires a little bit of work. We observe  from \eqref{eqn:boundedness}, that given ${\bf z}\in \R^{2,\circ}_+$, there exist  $k_1, k_2>0$ such $f_i({\bf z})\leq k_i$. Let $k=\max\set{k_1,k_2}$. Then we have for all ${\bf c}=(c_1,c_2)\in \R^{2,\circ}_+ $
 \[\frac{\sum_{i=1}^2c_iz_if_i({\bf z})}{1+\sum_{i=1}^2 c_iz_i}\leq k\frac{\sum_{i=1}^2c_iz_i}{\sum_{i=1}^2 c_iz_i}\leq k\;.\]
Similarly, 
\[1+\sum_{i=1}^2\rb{\abs{f_i({\bf z})+g_i^2({
\bf z})}}\leq 1+k+2=3+k\;.\]
Also, 
\[\sum_{i,j=1}^2\sigma_{ij}c_ic_jz_iz_jg_i({\bf z})g_j({\bf z})=\sum_{i=1}^2\sigma_i^2c_i^2z_i^2\;.\]
Hence $A_3$ will be satisfied if we can find $\gamma_b>0$ and ${\bf c}=(c_1,c_2)\in \R^{2,\circ}_+ $
such that 
\[k+\gamma_b(3+k)\leq \frac{1}{2} \frac{\sum_{i=1}^2\sigma_i^2c_i^2z_i^2}{\rb{1+\sum_{i=1}^2 c_iz_i}^2}\;.\]
We can use the Cauchy-Schwarz inequality to refine this condition further more. Indeed, put $\sigma=\min\set{\sigma_1,\sigma_2}$. Then by the Cauchy Schwarz inequality, 
\[ \rb{\sum_{i=1}^2 c_iz_i}^2\leq \sum_{i=1}^21^2\sum_{i=1}^2 c_i^2z_i^2=2\sum_{i=1}^2 c_i^2z_i^2\;.\]
Then $\ds \sum_{i=1}^2\sigma_i^2c_i^2z_i^2\geq \frac{1}{2}\sigma^2\rb{\sum_{i=1}^2 c_iz_i}^2$\;. Hence $A_3$ will be satisfied if we can find $\gamma_b>0$ and ${\bf c}=(c_1,c_2)\in \R^{2,\circ}_+ $
such that 
\[ k+\gamma_b(3+k)\leq \frac{\sigma}{4} \frac{\rb{\sum_{i=1}^2c_iz_i}^2}{\rb{1+\sum_{i=1}^2 c_iz_i}^2}\;.  \]
Now let $z_1, z_2\in \R_+$. We pick a  $\gamma_0 \in (0,1)$. There exists  $x\in \R^+$ such that $x\geq \frac{\gamma_0}{1-\gamma_0}$.  Therefore, we will have  $\frac{x}{1+x}\geq \gamma_0$ which implies $\frac{x^2}{(1+x)^2}\geq \gamma_0^2$. Choosing in particular $x$ in the interval  $[\min\set{z_2,z_2},\max\set{z_1,z_2}]$, there  exists $\bm{c}=(c_1,c_2) \in \R^{2,\circ}_+$ such that $x=\sum_{i=1}^2c_iz_i$ and $\ds \frac{\rb{\sum_{i=1}^2c_iz_i}^2}{(1+\sum_{i=1}^2 c_iz_i)^2}\geq \gamma_0^2$.
To finish, we choose $\gamma_b$ such that  $\ds \gamma_0^2=\frac{4(k+\gamma_b(3+k))}{\sigma}$ and the proof is complete.

 \end{document}